\documentclass[a4paper, amsfonts, amssymb, amsmath, reprint, showkeys, nofootinbib,superscriptaddress,twoside,pra]{revtex4-1}
\usepackage[english]{babel}
\usepackage[utf8]{inputenc}
\usepackage[colorinlistoftodos, color=green!40, prependcaption]{todonotes}
\usepackage{xcolor}
\usepackage{graphicx}
\usepackage{stackengine}
\usepackage[left=23mm,right=13mm,top=35mm,columnsep=15pt]{geometry} 
\usepackage{adjustbox}
\usepackage{placeins}
\usepackage[T1]{fontenc}
\usepackage{csquotes}
\usepackage[vcentermath]{youngtab}

\usepackage[nointegrals]{wasysym}
\usepackage{amsmath}
\usepackage{amssymb}
\usepackage{amsthm}
\usepackage{braket}
\usepackage{bm}
\usepackage{mathtools}
\usepackage{tensor}
\usepackage{mathrsfs}
\usepackage[vcentermath]{youngtab}

\def\openone{\leavevmode\hbox{\small$1$\normalsize\kern-.33em$1$}}

\newcommand{\abs}[1]{\left| #1 \right|} 
\newcommand{\avg}[1]{\left< #1 \right>}

\newcommand{\matrixel}[3]{\left< #1 \vphantom{#2#3} \right|
 #2 \left| #3 \vphantom{#1#2} \right>} 
\let\baraccent=\= % rename builtin command \= to \baraccent
\renewcommand{\=}[1]{\stackrel{#1}{=}}
\DeclareMathOperator{\Tr}{\text{tr}}
\newcommand{\ketbra}[2]{\left| #1 \vphantom{#2}\right>\!\!\left< #2\vphantom{#1}\right|}
\newcommand{\be}{\begin{equation}}
\newcommand{\bel}[1]{\begin{equation}\label{#1}}
\newcommand{\ee}{\end{equation}}

\newcommand{\C}{\mathbb{C}}

\renewcommand{\H}{\hat{\mathcal{H}}}

\renewcommand{\Tr}{\operatorname{tr}}

\newcommand\blfootnote[1]{%
  \begingroup
  \renewcommand\thefootnote{}\footnote{#1}%
  \addtocounter{footnote}{-1}%
  \endgroup
}

\graphicspath{ {./images/} }
\usepackage[pdftex, pdftitle={Article}, pdfauthor={Author}]{hyperref} % For hyperlinks in the PDF

\newtheorem{theorem}{Theorem}

\newtheorem{lemma}[theorem]{Lemma}

\begin{document}
\title{Generalized Collective Lamb Shift}
    \author{Andrew Stasiuk}\blfootnote{Andrew Stasiuk and Lane Gunderman are the co-first-authors of this work and contributed equally.}
    \email{andrew.stasiuk@uwaterloo.ca}
    \affiliation{The Institute for Quantum Computing, University of Waterloo, Waterloo, Ontario, N2L 3G1, Canada}
    \affiliation{Department of Applied Mathematics, University of Waterloo, Waterloo, Ontario, N2L 3G1, Canada}
    \author{Lane G. Gunderman}
    \email{lgunderman@uwaterloo.ca}
    \affiliation{The Institute for Quantum Computing, University of Waterloo, Waterloo, Ontario, N2L 3G1, Canada}
    \affiliation{Department of Physics and Astronomy, University of Waterloo, Waterloo, Ontario, N2L 3G1, Canada}
    \author{\nolinebreak Mohamed \nolinebreak El \nolinebreak Mandouh}
    \affiliation{The Institute for Quantum Computing, University of Waterloo, Waterloo, Ontario, N2L 3G1, Canada}
    \affiliation{Department of Applied Mathematics, University of Waterloo, Waterloo, Ontario, N2L 3G1, Canada}
    \author{Troy W. Borneman}
    \affiliation{The Institute for Quantum Computing, University of Waterloo, Waterloo, Ontario, N2L 3G1, Canada}
    \affiliation{High Q Technologies Inc, Waterloo, Ontario, N2L 3G1, Canada}
    \author{David G. Cory}
    \affiliation{The Institute for Quantum Computing, University of Waterloo, Waterloo, Ontario, N2L 3G1, Canada}
    \affiliation{Department of Chemistry, University of Waterloo, Waterloo, Ontario, N2L 3G1, Canada}
    %\affiliation{Canadian Institute for Advanced Research, Toronto, Ontario M5G 1Z8, Canada}
    %\affiliation{Perimeter Institute for Theoretical Physics, Waterloo, Ontario N2L 2Y5, Canada}

\date{\today} % Leave empty to omit a date

\begin{abstract}
Hybrid quantum systems consisting of an ensemble of two--level systems interacting with a single--mode electromagnetic field are important for the development of quantum information processors and other quantum devices. These systems are characterized by the set of energy level hybridizations, split by collective Lamb shifts, that occur when the ensemble and field mode interact coherently with high cooperativity. Computing the full set of Lamb shifts is generally intractable given the high dimensionality of many devices. In this work, we present a set of techniques that allow a compact description of the Lamb shift statistics across all collective angular momentum subspaces of the ensemble without using restrictive approximations on the state space. We use these techniques to both analyze the Lamb shift in all subspaces and excitation manifolds and to describe the average observed Lamb shift weighted over the degeneracies of all subspaces. 
\end{abstract}


\maketitle

%I believe leaving the sections in separate files is more organized, change it if you desire 

\section{Introduction} \label{sec:intro}
    The field of quantum electrodynamics (QED) was initiated by Lamb's discovery that an electron interacts with its own radiation field to split the energies of the $2s_{1/2}$ and $2p_{1/2}$ levels of the Hydrogen atom \cite{lamb1947fine}. This splitting, referred to as the Lamb shift, demonstrates that the electromagnetic field and vacuum are quantized. Cavity QED systems, where two--level quantum systems (such as atoms) are confined in a high-finesse cavity, also present features analogous to the Lamb shift. The light--matter interaction breaks degeneracies between separable field and atom states with the same number of excitations, $k$, hybridizing the states with a splitting that scales in magnitude as $\sqrt{k}$. Provided the cavity finesse is high enough and the atomic coherence long enough, this hybridization may be observed experimentally \cite{haroche2006exploring}. Understanding the structure of cavity QED Lamb shifts has become particularly important recently due to their importance in the development of large-scale quantum information processors \cite{blais_circuit_2020,fink_dressed_2009,yang_probing_2020,zou_implementation_2014} and hybrid quantum devices \cite{kurizki_quantum_2015,morton_hybrid_2011,xiang_hybrid_2013}, such as quantum memories for microwave photons \cite{kubo_hybrid_2011,grezes_multimode_2014,wu_storage_2010,zhu_coherent_2011} and optical photons \cite{hammerer_quantum_2010,afzelius_proposal_2013,vivoli_high-bandwidth_2013}. Collective Lamb shifts also play a crucial role in radiative ground-state cooling of an ensemble \cite{wood_cavity_2014,wood_cavity_2016,bienfait_controlling_2016,albanese_radiative_2020,ranjan_pulsed_2019}
    
    The notion of a Lamb shift in cavity QED may be formally defined using the Jaynes--Cummings Hamiltonian describing the light--matter interaction of a single two-level quantum system with a quantized single--mode electromagnetic field \cite{jaynes1963comparison} ($\hbar = 1$ throughout work):
    \begin{align}\label{eq.JC}
    \nonumber \H_{0} &= \omega_c \hat{a}^\dagger\hat{a} + \frac{\omega_s}{2}\hat{\sigma}_z,\\
    \nonumber \H_{int} &= g_0(\hat{a}^\dagger \hat{\sigma}_- + \hat{a}\hat{\sigma}_+),\\
    \H_{JC} &= \H_{0} + \H_{int},
    \end{align}
    where $\H_{0}$ is the Hamiltonian describing the quantization of the two--level system and single--mode separately, and $\H_{int}$ is the Hamiltonian describing the interaction of the two-level system with the field mode. The lowering (raising) operators, $\hat{a}$ ($\hat{a}^\dagger$), describe the annihilation (creation) of a photon in the field mode with energy $\omega_c$. The number operator, $\hat{a}^\dagger\hat{a}$, describes the quantization of the field mode in terms of the number of photons, $n$, occupying the mode, and defines the Fock eigenstates, $\ket{n}$, as $\hat{a}^\dagger\hat{a}\ket{n} = n\ket{n}$. The corresponding quantization of the two-level system is given by Zeeman eigenstates, $\hat{\sigma}_z\ket{\uparrow} = +\ket{\uparrow}$ and $\hat{\sigma}_z\ket{\downarrow} = -\ket{\downarrow}$, with energy splitting $\omega_s$ given by the Pauli $z$ spin operator. For brevity, we will refer to a general two-level quantum system as a spin and the single--mode electromagnetic field as a cavity.
    
    The spin--cavity interaction describes a coherent swapping of a single photon between the spin and the cavity mode, where $\hat{\sigma}_+$ and $\hat{\sigma}_-$ correspond, respectively, to the creation and annihilation of a photon in the spin system. The strength of the spin--cavity interaction is given by the geometric parameter
    \begin{equation}
    g_0 = g_e \mu_B \sqrt{\frac{\mu_0 \omega_c}{2 V_c}},    
    \end{equation}
    where $g_e$ is the electron Landau g-factor, $\mu_B$ is the Bohr magneton, $\mu_0$ is the permeability of free--space, and $V_c$ is the mode volume of the cavity. We restrict ourselves to the regime $g_0 \ll \omega_c, \omega_s$ such that a rotating--wave approximation (RWA) may be applied to suppress multi-photon processes. For simplicity, we will also restrict our argument to the case where the spin system and cavity mode are resonant ($\omega_c = \omega_s = \omega_0$). 
    
    We denote the separable $k$-excitation eigenstates of $\H_0$ as $\lbrace \ket{k}\ket{\downarrow}, \ket{k-1}\ket{\uparrow} \rbrace$ and, upon diagonalization under $\H_{int}$, the resulting hybridized spin-cavity eigenstates are $\lbrace \frac{1}{\sqrt{2}}(\ket{k}\ket{\downarrow}+\ket{k-1}\ket{\uparrow}),\frac{1}{\sqrt{2}}(\ket{k}\ket{\downarrow}-\ket{k-1}\ket{\uparrow})\rbrace$, with a Lamb shift splitting given by $g_0\sqrt{k}$. In the special case of $k=1$ (the single--excitation manifold), the Lamb shift splitting is commonly referred to as a ``normal mode'' or ``vacuum Rabi'' splitting \cite{haroche2006exploring}. The non--linearity of the Lamb shift with excitation number has been observed experimentally to verify the ``quantum'' nature of the spin--cavity interaction \cite{fink2008climbing}.
    
    The Jaynes--Cummings model may be generalized to an ensemble of $N$ non--interacting spins collectively interacting with a single--mode cavity to yield the Tavis--Cummings Hamiltonian \cite{tavis1968exact},
    \begin{equation}\label{eq.TC}
    \H_{TC} = \omega_0 (\hat{a}^\dagger\hat{a} + \hat{J}_z) +  g_0(\hat{a}^\dagger \hat{J}_- + \hat{a}\hat{J}_+),
    \end{equation}
    where the two--level spin operators have been replaced with collective operators that act identically over the ensemble of energetically indistinguishable spins and a RWA has once again been made. We will formally define the collective operators in the next section. An important feature of the TC Hamiltonian is that the spin--cavity interaction strength, $g_0$, is often replaced with an effective interaction strength that is enhanced by $\sqrt{N}$:
    \begin{equation}
    g_{eff} = g_0\sqrt{N}.
    \end{equation}
    This transformation is paired with a $1/\sqrt{N}$ term in the collective angular momentum raising and lowering operators, which we will omit in this work.
   
   The ensemble enhancement of the spin--cavity interaction strength has allowed observation of an analogous normal mode splitting (often referred to as ``strong coupling'') in ensemble spin systems interacting with high quality factor (high Q) cavities \cite{schuster_high-cooperativity_2010,benningshof_superconducting_2013,imamoglu_cavity_2009,kubo_strong_2010}. The relative strengths of the parameters necessary to resolve this splitting are formalized by defining the cooperativity:
    \begin{equation}
    C = \frac{4Ng_0^2 Q T_2}{\omega_0},    
    \end{equation}
    where $Q$ is the quality factor of the cavity and $T_2$ is the coherence time of the spin ensemble.
    
    In general, experimentally observed splittings in a high--cooperativity spin--cavity system are a complex function of the many Lamb shifts that occur in various collective angular momentum subspaces and excitation subspaces. Calculating the full set of Lamb shifts is generally intractable for systems of the size required to build useful quantum devices, leading to a number of approximation methods being utilized to analyze experimental data. The most common approximations are to restrict the treatment to only the largest, permutation--invariant, Dicke subspace, or to treat the spin ensemble in the low--excitation regime as a simple quantum harmonic oscillator \cite{holstein1940field}. Many limiting results have been shown and can be seen in \cite{garraway2011dicke}, however, we analyze many aspects of the TC Hamiltonian in further detail. 
  
     In this paper, we revisit the structure of the Tavis--Cummings Hamiltonian and find that we are able to identify a two parameter family of collective Lamb shifts indexed by total number of excitations and collective angular momentum. We then show that, in general, a correct description of the energy landscape requires considering many collective angular momentum subspaces in all but the simplest limiting cases. We show that the representative set of spaces never limit to the Dicke space, nor a constant value, but instead grow as $O\left(\sqrt{N}\right)$. We then proceed to give non--trivial descriptive statistics on these Lamb shifts' scaling behaviors, culminating in a description of these statistics upon averaging over the degeneracies of these collective angular momentum subspaces. These insights provide bounds and estimates for structures that can be experimentally observed inthe near future.
     
     This paper is structured as follows. We begin by recalling common definitions in section \ref{sec:def}. Then, in section \ref{sec:methods} we discuss features of the Tavis--Cummings Hamiltonian that allow us to break it into subspaces of constant unperturbed energy and define degenerate copies of coupling matrices which act within these subspaces. We also show some examples of these coupling matrices as well as introduce our generalized collective Lamb shift. In section \ref{sec:results} we show our primary results, which include finding the maximally degenerate angular momentum subspace, and provide proof that the majority of the relevant dynamics occurs in this subspace and neighboring subspaces. We also describe basic properties of the collective Lamb shift within each subspace, and lastly combining these collective Lamb shifts over their degeneracies to predict what should be experimentally seen for such a system at a given number of total excitations. We then end with a discussion of our results in sections \ref{sec:discussion} and \ref{sec:conclusion}.  
\section{Definitions} \label{sec:def}
    We include our choice of notation and definitions in this section. The Pauli operators are written in the Zeeman basis, so that
    \begin{equation}
        \hat{\sigma}_z = \ketbra{\uparrow}{\uparrow} - \ketbra{\downarrow}{\downarrow}.
    \end{equation}
    Also,
    \begin{equation}
        \hat{\sigma}_+ \ket{\downarrow} = \ket{\uparrow},\quad \hat{\sigma}_+ |\uparrow\rangle=0 .
    \end{equation}
    These spin operators can be combined as a sum of tensor products to produce the collective versions of these spin operators. Let $N$ be the number of spin--$1/2$ particles. Then, the collective $z$ angular momentum operator is given as,
    \begin{equation}
        \hat{J}_z = \frac{1}{2} \sum_{i=1}^N \hat{\sigma}_z^{(i)},
    \end{equation}
    where the superscript on the Pauli operator indicates action only on the $i$-th particle. Similarly, the collective raising and lowering collective angular momentum operators are given as
    \begin{equation}
        \hat{J}_\pm = \sum_{i=1}^N \hat{\sigma}_\pm^{(i)}.
    \end{equation}
    The collective operators span the $\mathfrak{sl}(2;\C)$ Lie algebra, and thus satisfy the following commutation relations:
    \begin{align}
        [\hat{J}_z , \hat{J}_\pm] = \pm \hat{J}_\pm,\\
        [\hat{J}_+, \hat{J}_-] = 2\hat{J}_z.
    \end{align}
    The collective operator algebra is a sub-algebra of self--adjoint operators acting on the $N$-spin system, and conveniently satisfies the same commutation relations as those for a single particle spin operator. By a change of basis, we can identify the transverse spin operators,
    \begin{align}
        \hat{J}_x &= \frac{1}{2}\big(\hat{J}_+ + \hat{J}_-\big),\\
        \hat{J}_y &= \frac{1}{2i}\big(\hat{J}_+ - \hat{J}_-\big).
    \end{align}
    The transverse spin operators, along with $\hat{J}_z$, span a collective $\mathfrak{su}(2)$ algebra, which differ from a single spin-1/2 Pauli operators in that the collective operators are not involutory. The representations of the $\mathfrak{sl}(2;\C)$ operators can be defined by their action on a state of total angular momentum $j$ with $z$ component $m$:
    \begin{align}
        \hat{J}_z \ket{j,m} &= m\ket{j,m}\\
        \hat{J}_\pm \ket{j,m} &= \sqrt{j(j+1) - m(m\pm 1)}\ket{j,m \pm 1}.
    \end{align}
    
    Throughout this work, we focus on two good quantum numbers representing conserved quantities. The first of these is the total angular momentum, $j$, which determines the eigenvalues of the total angular momentum operator,
    \begin{equation}
        \hat{\bm{J}}^2 = \hat{J}_x^2 + \hat{J}_y^2 + \hat{J}_z^2, 
    \end{equation}
    with eigenvalues $j(j+1)$. The second of these conserved quantities is the number of total excitations, $k$, given as the eigenvalues of the excitation operator,
    \begin{equation}
        \hat{K} = \hat{a}^\dagger\hat{a} + \hat{J}_z + \frac{N}{2}\openone.
    \end{equation}
    The scaled identity term in the excitation operator ensures excitations are non-negative, as the action of $\hat{J}_z$ on the ground state has eigenvalue $-N/2$.
    
    All of the Hamiltonians considered in this work share the same internal structure, defined by 
    \begin{equation}
        \H_0 =  \omega_0 (\hat{a}^\dagger\hat{a} + \hat{J}_z).
    \end{equation}
    The interaction Hamiltonian is given by the Dicke model \cite{dicke1954coherence}:
    \begin{equation}\label{dicke}
        \H_{D,int} = 2g_0 \big(\hat{a} + \hat{a}^\dagger\big) \hat{J}_x.
    \end{equation}
    Applying the rotating wave approximation, the counter--rotating term is discarded, leaving us with
    \begin{equation}
        \H_{int} = g_0 \big(\hat{a} \hat{J}_+ + \hat{a}^\dagger \hat{J}_- \big),
    \end{equation}
    which is the interaction term in the Tavis--Cummings (TC) Hamiltonian. \if{false}Note that tensor products are omitted for brevity, so that, for example,
    \begin{equation}
        \hat{a} \otimes \hat{J}_+ \equiv \hat{a}\hat{J}_+.
    \end{equation}\fi

    Lastly, we note that all states live within the Hilbert space,
    \begin{equation}
        \mathscr{H} = \operatorname{L}^2\left(\mathbb{R}\right) \otimes \big(\C^2\big)^{\otimes N}.
    \end{equation}
    
    It is common to perform a Holstein--Primakoff transformation on the collective angular momentum operators in order to simplify the underlying algebra \cite{holstein1940field}. This transformation is valid on a single subspace of angular momentum $j$, such that
    \begin{align}
        \hat{J}_+ &\longrightarrow \hat{b}^\dagger \sqrt{2j\openone - \hat{b}^\dagger\hat{b}}\\
        \hat{J}_- &\longrightarrow  \sqrt{2j\openone - \hat{b}^\dagger\hat{b}}\hspace{.25cm}\hat{b}.
    \end{align}
   By requiring that standard angular momentum commutation relations are maintained, the transformation for $\hat{J}_z$ is then fixed:
   \begin{equation}
       \hat{J}_z \longrightarrow \hat{b}^\dagger \hat{b} - j \openone.
   \end{equation}
   Usually $j$ is taken to be the Dicke subspace, such that $j=N/2$, and $N$ is assumed to be large compared to the number of excitations. If the number of excitations approach $j$, the spin system begins to saturate and the approximation becomes increasingly invalid \cite{ressayre1975holstein}. In general, thermal population of the Dicke space is negligible at nearly all experimental temperatures \cite{wesenberg2002mixed}, so we avoid making restrictive approximations in this work and treat the Hamiltonian in generality across all subspaces and excitation manifolds.

    We are now equipped with our primary definitions and notations. In the next section we discuss what features of our Hamiltonian allow us to decompose the interaction portion of the Hamiltonian into a direct sum of coupling matrices, allowing us to show results regarding the subspaces forming the bases of these coupling matrices.
\section{Symmetry and Subspace Decomposition} \label{sec:methods}
    In this section, we discuss the symmetries of various models of spin ensembles interacting with a cavity. Through the use of conserved quantities, we motivate a decomposition of the TC Hamiltonian into a two parameter family of subspaces indexed by the good quantum numbers present in the system. We then solidify the remaining notation to be used in the rest of the paper, relying heavily on the symmetry motivated subspace decomposition, and provide instructive examples for small values of $N$.

\subsection{Symmetries of Light--Matter Interaction} \label{subsec:Symmetry}

    Generally, an ensemble of $N$ spins identically coupled to a single cavity mode is described by the Dicke model, with a Hamiltonian given by 
    \begin{align}
        \mathcal{\hat{H}} =  \omega_{0} (\hat{a}^\dagger \hat{a} + \hat{J}_z) + g_0(\hat{a}^\dagger + \hat{a} ) \hat{J}_x.
    \end{align}
    The Dicke Hamiltonian can be decomposed into two distinct parts, the bare spin and cavity energies, $\mathcal{\hat{H}}_0$, and the spin--cavity interaction $\mathcal{\hat{H}}_{\text{D,int}}$. When the collective spin--cavity interaction, $g_{eff}$, is zero, the ground state is $\ket{0}\ket{N/2, -N/2}$, which represents the state with zero photons in the cavity and all spins in their ground states. When $g_{eff}>0$, the Dicke Hamiltonian is symmetric under the parity operator $\operatorname{\hat{\Pi}} = \operatorname{exp}\left[-i \pi \left(\hat{J}_z + \hat{a}^\dagger\hat{a}\right)\right]$, with eigenvalues $\pm 1$. This implies that the Hilbert space of the Dicke model can be decomposed into a direct sum of two spaces labelled by the parity operator's sign: $\mathscr{H} = \mathscr{H}_{+} \oplus \mathscr{H}_{-}$. In this model, excitations are not conserved, and the two parity subspaces are infinite dimensional.
    
    For the case of $N=1$, this particularization of the Dicke model is known as the Quantum Rabi Model (QRM), which has recently been solved \cite{zhong2013analytical,maciejewski2014full}, where eigenvalues and eigenstates are given in terms of special functions \cite{judd1979exact}. The existence of this solution can be seen directly from the symmetry group of the Hamiltonian, as the parity symmetry is sufficient to show that the QRM is integrable \cite{braak2011integrability}. When we consider $N>1$, the parity symmetry is no longer sufficient to show integrability, and it is expected that the Dicke model is not exactly solvable \cite{baxter2016exactly}; in other words, there are no explicit solutions in terms of any known functions.
    
    Turning now to the model of interest for this paper, the Tavis--Cummings Hamiltonian is derived by the application of a RWA to the Dicke Hamiltonian. Unlike the Dicke Model, the TC model admits a continuous symmetry described by the circle group, $U(1)$, in addition to parity symmetry and total angular--momentum symmetry. The generator of the continuous symmetry has infinite eigenvalues, enumerated by  $k \in \mathbb{N}$, while the total angular momentum symmetry has eigenvalues $j = N/2, N/2-1, \cdots, 1/2\ (0)$, where the last value for $j$ is determined by whether $N$ is odd or even.
    The additional symmetry is sufficient to make the Tavis--Cummings model integrable and solvable, which is supported by the Bethe ansatz solution provided by Bogoliubov \cite{bogoliubov1996exact, bogolyubov2000algebraic}. 
    
    Given that $\hat{J}_+$ conserves total angular momentum $j$, the repeated action of $\hat{J}_+$ on the ground state of an $N$ spin ensemble will only populate the $N+1$ fully symmetric states in the Dicke subspace. Bogoliubov utilized these orbits to verify a Bethe ansatz solution of the Tavis--Cummings model is correct, casting the eigenvalue problem as equivalent to solving a differential equation \cite{bogoliubov1996exact, bogolyubov2000algebraic}. Translating their construction into our notation, the primary expression is:
    \begin{equation}\label{eq:bogolub_construction}
        \ket{\Phi_{j,k}^\lambda} = \sum_{m=0}^k A_{j,k,m}^\lambda (\hat{a}^\dagger)^{k-m} \hat{J}_+^m \ket{0} \ket{j,-j},
    \end{equation}
    for recursively defined scalar coefficients $A_{j,k,m}^\lambda$ determined from difference equations, where $j$ indicates the angular momentum space, $k$ the excitation subspace, and $\lambda$ to a labeled eigenvector within the $(j,k)$ subspace. Putting together these symmetry observations, we see that the TC Hamiltonian can be tractably analyzed in terms of its structure and dynamics. In section \ref{sec:results}, we provide a detailed analysis of the TC Hamiltonian's energy level structure across all non-interacting subspaces. 
    
    The two symmetries of the TC model directly imply that the Hamiltonian admits a two parameter subspace decomposition. We will repeatedly make use of this fact throughout the remainder of our analysis. Within the context of previous work, the Holstein--Primakoff approximation largely ignores the second parameter $j$ by focusing on a single value of it, particularly the $j=N/2$ subspace which is being treated as a single harmonic oscillator. The Bogoliubov solution via Bethe ansatz, while correct, is equally as hard as solving the eigenvalue problem itself. Further work attempting to directly analyze large photon number behavior via a direct diagonalization approach has been performed by restriction to the Dicke subspace and tested experimentally by Chiorescu \textit{et al} \cite{chiorescu2010magnetic}. We will demonstrate in later sections that the most dominantly contributing angular momentum subspaces are, in general, those with the lowest $O(\sqrt{N})$ $j$ values allowed by the model.
    
\subsection{Subspace Decomposition of the TC Model}

    Subsection \ref{subsec:Symmetry} argued that we can use group theory to decompose the total Hamiltonian into a direct sum structure indexed by two parameters defined by the conserved quantities of the system, $j$ and $k$. A direct sum decomposition is not novel, and was given explicitly in the original 1968 paper defining the Tavis--Cummings Hamiltonian \cite{tavis1968exact}. Recast in our notation, the decomposition is
    \begin{equation}
        \H \cong \bigoplus_{j,k}\big( \omega_0 k \openone_{j,k} +  g_0 L(j,k)\big),
    \end{equation}
    where $L(j,k)$ are a natural representation of the interaction Hamiltonian, which we define in equation \eqref{eq:coupling_matrix} and refer to as the coupling matrices. While Tavis and Cummings focused on the eigenstates of their model, computed by recasting the diagonalization problem as a differential equation \cite{tavis1968exact}, our work focuses on the energy eigenvalue problem, utilizing modern insights into numerical linear algebra to provide a deeper analysis.
    
    We define a natural basis for a general $(j,k)$ subspace with total angular momentum $j$ and $k$ excitations as 
    \begin{equation}\label{eq:lamb_basis}
        \mathcal{B}_{j,k} = \lbrace \ket{\alpha_{j,k}} \,|\, \alpha = 1,\cdots,n_{j,k}, n_{j,k} + 1\rbrace,
    \end{equation}
    using a shorthand ket representation of the tensor product of a spin-cavity state
    \begin{equation}
        \ket{\alpha_{j,k}} = \ket{k-\alpha-k_0(j)}\ket{j,-j+\alpha}.
    \end{equation}
    The single parameter, $\alpha$, provides a convenient representation of states within a $(j,k)$ subspace. The value, $n_{j,k} = \abs{ \mathcal{B}_{j,k} } - 1$, one less than the dimension, is chosen for convenience. We define $k_0(j) = N/2-j$ as the number of excitations present in the ground state of an angular momentum $j$ subspace within an $N$ spin ensemble. Explicitly, $n_{j,k}$ is given as
    \begin{equation}
        n_{j,k} = \operatorname{min}\lbrace 2j, k - k_0(j)\rbrace.
    \end{equation}
    If $k < k_0(j)$, then the basis set is empty and there are no states present at this excitation level within the $j$ angular momentum subspace.
    
    Under unitary evolution generated by collective operators, two subspaces with the same value of $j$ stemming from ensembles of differing $N$ will behave identically, as these subspaces have isomorphic representations. The main functional difference between them is their relative locations within the energy level spectrum of their respective Hamiltonians. Thus, while the evolution or action of a collective operator can be calculated identically, the resultant contribution of the evolution to aggregate statistics or an observable will be weighted differently.
    
    The coupling matrices' entries can be found by applying the interaction term from $\mathcal{\hat{H}}_{TC}$ on the bases defined in equation \eqref{eq:lamb_basis}. The Lamb shift coupling matrix for the $(j,k)$ subspace is then given by
    \begin{align}\label{eq:coupling_matrix}
        \nonumber L(j,k) &= \sum_{\alpha=1}^{n_{j,k}} l_{\alpha}(j,k)\bigg(\ketbra{\alpha_{j,k}}{(\alpha+1)_{j,k}}\\
        &\hspace{1.5cm}+ \ketbra{(\alpha+1)_{j,k}}{\alpha_{j,k}}\bigg),
    \end{align}
    where the matrix elements $l_{\alpha}(j,k)$ are given by
    \begin{align}
        \nonumber &\frac{1}{ g_0}\matrixel{\alpha_{j,k}}{\H_{int}}{(\alpha+1)_{j,k}}\\
        &= \sqrt{\big(2\alpha j - \alpha(\alpha-1)
        \big)\big(k-k_0(j)-\alpha+1\big)}.
    \end{align}
    In the above expression subscripts are only included within kets such as $\ket{\alpha_{j,k}}$, while $\alpha$ itself is a scalar.
    
    The index $j$ runs from $N/2$ to $0$ ($1/2$) when $N$ is even (odd). Each angular momentum space is of dimension $2j+1$, and so the total number of spin states accounted for across all values of $j$ is $O(N^2)$,\if{false}with the leading term given as $\frac{N^2}{4}$, \fi which is far less than the full space's dimension of $2^N$. As of this point we have neglected to include the degeneracy of each of the angular momentum subspaces. The degeneracy of the subspace with total angular momentum $j$ on $N$ spins is given as
    \begin{equation}
        d_j = \frac{N! (2j+1)}{(N/2 - j)!(N/2+j+1)!}.
    \end{equation}
    That is, there are $d_j$ disjoint angular momentum subspaces with total angular momentum $j$ present in a direct sum decomposition of $\big(\C^2\big)^{\otimes N}$ \cite{wesenberg2002mixed}. By including this degeneracy we recover the identity that the sum over the dimension of all disjoint subspaces is equal to the dimension of the entire space,
    \begin{equation}
        \sum_j (2j+1)d_j=2^N.
    \end{equation}
    
    Through Schur--Weyl duality, we can associate total angular momentum symmetry with invariance over permutations (or subgroups of permutations) of the ordering of the underlying spin Hilbert spaces \cite{weyl1946classical}. Within this context, the subspace with $j=N/2$, commonly known as the Dicke subspace, is referred to as the fully symmetric subspace. This is due to every angular momentum state within the $j=N/2$ subspace remaining invariant under action of any permutation in $S(N)$, the permutation group of order $N$. The remaining subspaces have a more complex structure under the action of a spin-permutation. 
    
    Importantly, each degenerate copy of a $j$ subspace can be naturally and uniquely labelled by a Young Tableau. If one wished to consider a perturbation to the TC Hamiltonian which distinguished individual spins, such as a field inhomogeneity, then these Young Tableaux would be required to properly determine the perturbation's action on subspaces with identical total angular momentum $j$. As previously mentioned, in our work we focus on the ideal case with no perturbations. As such, it is sufficient to treat each degenerate copy of a given total angular momentum subspace as identical. Under this identification, we are able to reduce the effective spin dimensionality from $2^N$ to $O(N^2)$.

\subsection{Examples for Small $N$}\label{examples}

We now proceed to explicitly calculate the collective Lamb shift in two small $N$ spin--cavity systems. This is shown mathematically by re--diagonalization under the perturbative interaction Hamiltonian and finding how these re--diagonalized states' energies differ from those where $\mathcal{H}_{int}=0$, or equivalently, where $g_0 = 0$.

% As elsewhere will use the convention of our first register, as a non--negative number, being the photon excitation number, while the second register is for the spin ensemble states--indicated with $|\uparrow\rangle$ and $|\downarrow\rangle$ for each spin.

\subsubsection{Single Spin}

The case of a single spin coupled to one electromagnetic mode is known as the Jaynes--Cummings model \cite{jaynes1963comparison}. The JC Hamiltonian follows from the application of the RWA on the Dicke Hamiltonian for a single spin, also known as the Rabi model, and is given by 
\begin{equation}
    \H_0 + \H_{int} =  \omega_{0} (\hat{a}^\dag \hat{a} + \hat{\sigma}_z) + g_0 \left( \hat{a}^\dag \sigma_{-} +\hat{a} \sigma_{+} \right).
\end{equation}
Since $\H_{int}$ couples spins with equal energy in the unperturbed spectrum, the Hilbert space decouples into blocks of constant total excitation, indexed by the good quantum number $k$:
\begin{multline}
 \bigoplus_k |\psi_k\rangle\Big[ \langle\psi_k| \mathcal{\hat{H}}_0+\mathcal{\hat{H}}_{int} |\phi_k\rangle\Big] \langle\phi_k|,\\
 \text{ with }\mathcal{\hat{H}}_0\phi_k=E_k\phi_k \text{ and } \mathcal{\hat{H}}_0\psi_k=E_k\psi_k. %,\psi\in \text{basis}(H_0) \}
\end{multline}

The ground state $|0\rangle|\downarrow\rangle$ is unique, and is thus not hybridized. For the remaining states, we utilize the fact that excitations are conserved. Consider the two states with excitations $k > 0$, defined by $\lbrace \ket{k}\ket{\downarrow}, \ket{k-1}\ket{\uparrow}\rbrace$. The interaction Hamiltonian represented in this basis is given by the direct sum over all two--dimensional excitation spaces as follows:
\begin{equation}
\mathcal{\hat{H}}_{int}\cong \bigoplus_k g_0\begin{bmatrix}
0 & \sqrt{k}\\
\sqrt{k} & 0
\end{bmatrix}.
\end{equation}
The $k$ excitation representation of the interaction Hamiltonian has energy eigenvalues given by
\begin{equation}
    E_{k,\pm}= k\omega_0\pm g_0\sqrt{k},
\end{equation}
which correspond to the following energy eigenstates:
\begin{equation}
    |\psi_{k,\pm}\rangle=|k\rangle|\downarrow\rangle\pm |k-1\rangle|\uparrow\rangle.
\end{equation}

\subsubsection{Three Spins}\label{N3}

We now consider an $N=3$ spin--cavity system. We demonstrate the utility of the subspace decomposition technique by solving for the eigenstructure exactly. When the number of excitations are such that $k \leq 2$, the number of hybridized states are sub-maximal, as illustrated in figure \ref{fig:TC_hybrid}. For the purpose of this example, we focus on solving for a general collection of excitation subspaces with $k \geq 3$, ensuring that all $2^3=8$ spin states participate in hybridization. For completeness, we provide the solutions to the $N=3$ spin model with $k < 3$ excitations, as well as the $N=2$ spin model in the appendix using the same techniques illustrated in this section. For $N=3$ and $k\geq3$, a matrix representation of the interaction Hamiltonian is given by the matrix $L(k)$ in equation \eqref{eq:3spin_coupling}.

\begin{widetext}
\begin{equation}\label{eq:3spin_coupling}
L(k) = 
\begin{bmatrix}
0 & \sqrt{k} & \sqrt{k} & 0 & \sqrt{k} & 0 & 0 & 0\\
\sqrt{k} & 0 & 0 & \sqrt{k-1} & 0 & \sqrt{k-1} & 0 & 0\\
\sqrt{k} & 0 & 0 & \sqrt{k-1} & 0 & 0 & \sqrt{k-1} & 0\\
0 & \sqrt{k-1} & \sqrt{k-1} & 0 & 0 & 0 & 0 & \sqrt{k-2}\\
\sqrt{k} & 0 & 0 & 0 & 0 & \sqrt{k-1} & \sqrt{k-1} & 0\\
0 & \sqrt{k-1} & 0 & 0 & \sqrt{k-1} & 0 & 0 & \sqrt{k-2}\\
0 & 0 & \sqrt{k-1} & 0 & \sqrt{k-1} & 0 & 0 & \sqrt{k-2}\\
0 & 0 & 0 & \sqrt{k-2} & 0 & \sqrt{k-2} & \sqrt{k-2} & 0
\end{bmatrix},
\end{equation}
where the ordered basis states for this matrix representation are given by the set $\{|k\rangle |\downarrow\downarrow\downarrow\rangle,|k-1\rangle |\downarrow\downarrow\uparrow\rangle,\ldots |k-3\rangle |\uparrow\uparrow\uparrow\rangle\}$. This can be decomposed into 3 distinct subspaces, $\frac{3}{2} \oplus \frac{1}{2}\oplus \frac{1}{2}$, as follows:
\begin{equation}
    \begin{bmatrix}
    0 & \sqrt{3}\sqrt{k} & 0 & 0\\
    \sqrt{3}\sqrt{k} & 0 & 2\sqrt{k-1} & 0\\
    0 & 2\sqrt{k-1} & 0 & \sqrt{3}\sqrt{k-2}\\
    0 & 0 & \sqrt{3}\sqrt{k-2} & 0
    \end{bmatrix}\oplus\begin{bmatrix}
    0 & \sqrt{k-1}\\
    \sqrt{k-1} & 0
    \end{bmatrix}\oplus\begin{bmatrix}
    0 & \sqrt{k-1}\\
    \sqrt{k-1} & 0
    \end{bmatrix}.
\end{equation}
The first matrix is written in the Dicke (fully symmetric) basis (normalized versions of $|k-m\rangle \hat{J}_+^m|\downarrow\downarrow\downarrow\rangle$ with $m\in \{0,1,2,3\}$), while the second and third are written in terms of the composite spin--1/2 bases, given by
\begin{eqnarray}
\frac{1}{\sqrt{2}}|k-1\rangle(|\downarrow\uparrow\downarrow\rangle-|\uparrow\downarrow\downarrow\rangle)&,&\quad \frac{1}{\sqrt{2}}|k-2\rangle (|\downarrow\uparrow\uparrow\rangle-|\uparrow\downarrow\uparrow\rangle)\\
\frac{1}{\sqrt{6}}|k-1\rangle(2|\downarrow\downarrow\uparrow\rangle-|\uparrow\downarrow\downarrow\rangle-|\downarrow\uparrow\downarrow\rangle)&,&\quad \frac{1}{\sqrt{6}} |k-2\rangle (|\uparrow\downarrow\uparrow\rangle+|\downarrow\uparrow\uparrow\rangle-2|\uparrow\uparrow\downarrow\rangle).
\end{eqnarray}
\if{false}
\begin{eqnarray}
    |k-2\rangle |1/2_{2\text{ excitations},\text{ label }0}^{3\text{ spins}}\rangle&=&\frac{1}{\sqrt{3}}[(\frac{1}{\sqrt{2}}|\uparrow\downarrow\rangle+\frac{1}{\sqrt{2}}|\downarrow\uparrow\rangle)|\uparrow\rangle-\sqrt{2}|\uparrow\uparrow\downarrow\rangle]\\
    |k-1\rangle |1/2_{1,0}^3\rangle&=&\frac{1}{\sqrt{3}}[\frac{1}{\sqrt{2}}(2|\downarrow\downarrow\uparrow\rangle-|\uparrow\downarrow\downarrow\rangle-|\downarrow\uparrow\downarrow\rangle)]\\
    |k-2\rangle |1/2_{2,1}^3\rangle&=& \frac{1}{\sqrt{2}}\left[|\downarrow\uparrow\uparrow\rangle-|\uparrow\downarrow\uparrow\rangle\right]\\ 
    |k-1\rangle |1/2_{1,1}^3\rangle &=&\frac{1}{\sqrt{2}}[|\downarrow\uparrow\downarrow\rangle-|\uparrow\downarrow\downarrow\rangle].
\end{eqnarray}
\fi
The matrix representations for the degenerate spin-1/2 subspaces are identical, and thus indistinguishable under a collective operation or measurement. We have freedom in the choices of bases for the degenerate subspaces; the states we give are the standard basis states for these subspaces, as computed via a Clebsch--Gordon table \cite{mann2011introduction}.

We diagonalize each block individually starting with the matrix representing the $j=3/2$ subspace, and find the resulting (non-normalized) Lamb-shifted dressed states are given by the following superpositions:
\begin{eqnarray}\label{unevensplit}
    \nonumber |3/2, k; \pm_1, \pm_2 \rangle &:=& \pm_1\sqrt{5k-5\mp_2 \sqrt{25-32k+16k^2}}(2k-5\pm_2 \sqrt{25-32k+16k^2})|k\rangle|\downarrow\downarrow\downarrow\rangle\\
    \nonumber &&+
    (1+2k\mp_2 \sqrt{25-32k+16k^2})\sqrt{3}\sqrt{k}|k-1\rangle\frac{1}{\sqrt{3}}(|\downarrow\downarrow\uparrow\rangle+|\downarrow\uparrow\downarrow\rangle+|\uparrow\downarrow\downarrow\rangle)\\
    \nonumber &&\mp_1 2\sqrt{5k-5\mp_2 \sqrt{25-32k+16k^2}}\sqrt{3}\sqrt{k-1}\sqrt{k}|k-2\rangle\frac{1}{\sqrt{3}}(|\downarrow\uparrow\uparrow\rangle+|\uparrow\downarrow\uparrow\rangle+|\uparrow\uparrow\downarrow\rangle)\\
    &&+6\sqrt{k-2}\sqrt{k-1}\sqrt{k}|k-3\rangle |\uparrow\uparrow\uparrow\rangle,
\end{eqnarray}
each with associated energy
\begin{equation}
    E_{k;\pm_1\pm_2}=k\omega_0 \mp_1 g_0 \sqrt{5k-5\mp_2\sqrt{16k^2-32k+25}}.
\end{equation}
We have introduced a shorthand notation via a subscript on the $\pm$ sign, such that $\pm_1,\pm_2$ are a pair of sign choices (and $\mp_1$ indicates that the opposite sign as $\pm_1$ is used, and likewise for $\mp_2$) which allows for a more compact expression for all four dressed states. The four perturbed energy values are not equally spaced, though they still come in oppositely signed pairs of equal magnitude. We show in a later section that the eigenvalues always come in oppositely signed pairs.

For the remaining two matrices with $j=\frac{1}{2}$ in the direct sum decomposition, we note that these systems are algebraically equivalent to the single spin model. This equivalence allows us to immediately write down the diagonalized states and perturbed energies:
\begin{align*}
|1/2, k; \pm\rangle_1 &:= \frac{1}{2}[|k-1\rangle[|\downarrow\uparrow\downarrow\rangle-|\uparrow\downarrow\downarrow\rangle]\pm|k-2\rangle [|\downarrow\uparrow\uparrow\rangle-|\uparrow\downarrow\uparrow\rangle]],\\
E_{k;\pm} &=k\omega_0 \pm g_0\sqrt{k-1}\\
|1/2,\pm\rangle_2&:= \frac{1}{2\sqrt{3}}[|k-1\rangle[2|\downarrow\downarrow\uparrow\rangle-|\uparrow\downarrow\downarrow\rangle-|\downarrow\uparrow\downarrow\rangle]\pm |k-2\rangle [|\uparrow\downarrow\uparrow\rangle+|\downarrow\uparrow\uparrow\rangle-2|\uparrow\uparrow\downarrow\rangle]],\\
E_{k;\pm} &= k\omega_0 \pm g_0\sqrt{k-1}.
\end{align*}
The subscript on the kets in the above equations indicate the arbitrarily chosen degeneracy label of that subspace. This provides the full spectrum and dressed states for $k\geq 3$. Figure \ref{fig:TC_hybrid} illustrates the energy level diagram of the $N=3$ example, showing hybridization for the $0\leq k \leq 3$ subspaces, as well as collective dipole allowed transitions between dressed states.

\begin{figure}
    \centering
    \includegraphics{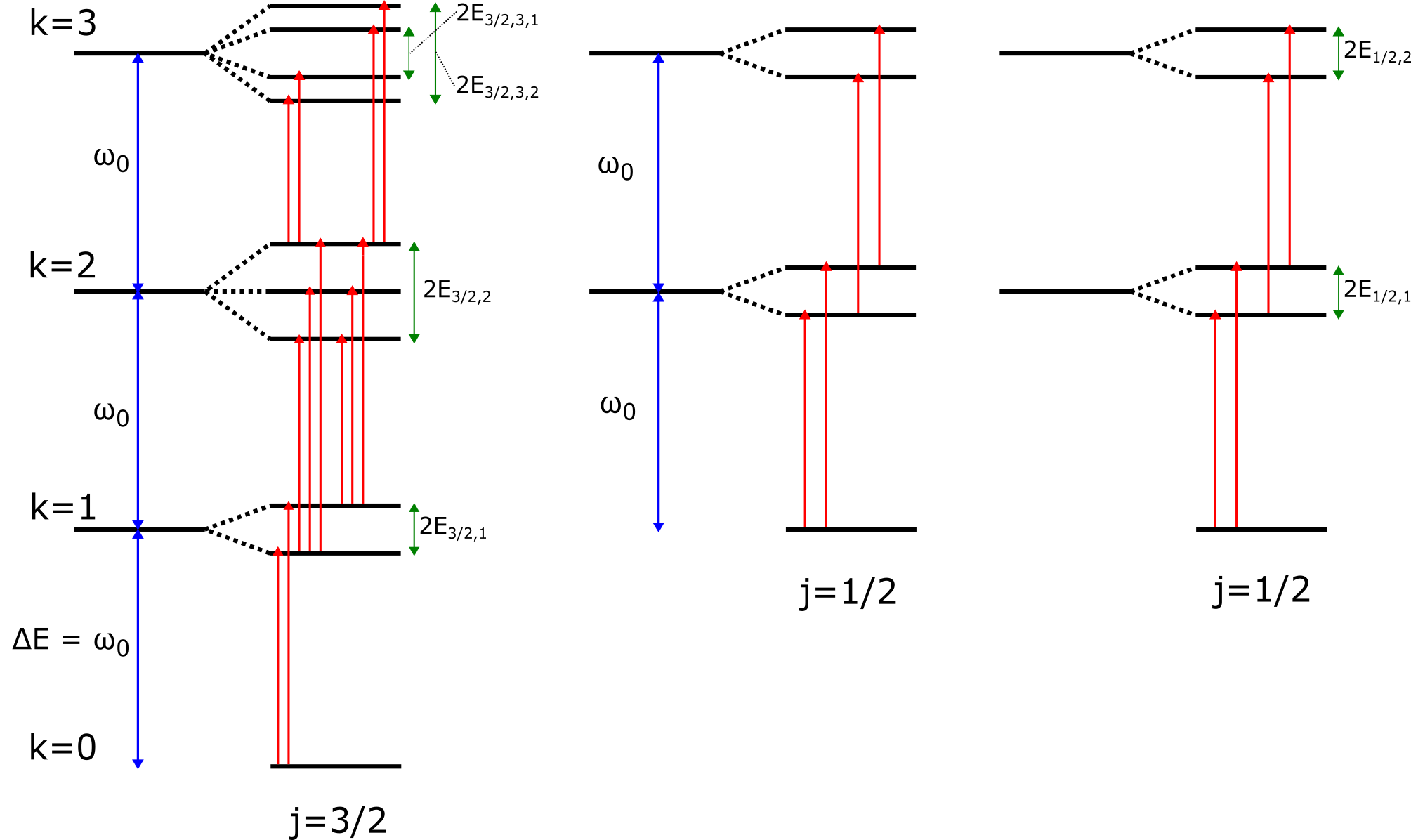}
    \caption{Illustration of the resulting hybridization of energy levels in the Tavis--Cummings model for $N=3$, explicitly on resonance such that $\omega_0 = \omega_s  = \omega_c$. Vertical single arrow lines (red) indicate transitions mediated by $\hat{J}_+$, meaning that the eigenstates represented by the horizontal bars have a non-zero $\hat{J}_+$ matrix element. Transitions are all--to--all between neighboring excitation subspaces of the same angular momentum, with some transitions between the $k=2$ and $k=3$ subspaces omitted for clarity. Note that there are no allowed transitions via collective spin or photon operators between distinct angular momentum subspaces, regardless of the value of $j$. Separation between excitation spaces is a constant $\omega_0$, denoted by bidirectional arrows (blue) between the pre--hybridized angular momentum states. Lamb shift splittings are denoted be bidirectional arrows (green) to the right of the hybridized states. In the $j=1/2$ subspaces, these splittings are given by $E_{1/2,k} = g_0 \sqrt{k}$. In the $j=3/2$ subspace, the Lamb shifts are given by: $E_{3/2,1} = g_0 \sqrt{3} \approx 1.73 g_0$, $E_{3/2,2} = g_0 \sqrt{10} \approx 3.16 g_0$, $E_{3/2,3,1} = g_0\sqrt{10 - \sqrt{73}} \approx 1.21 g_0$, and $E_{3/2,3,2} = g_0\sqrt{10 + \sqrt{73}} \approx 4.31 g_0$.}
    \label{fig:TC_hybrid}
\end{figure}
\end{widetext}

Before we move to the general case, we remark on a few well-known aspects of the solutions provided for $N=1,3$. Firstly, through the direct sum decomposition we see that the model decomposes into subspaces which are disconnected under the interaction portion of the total Hamiltonian. Secondly, through this decomposition, computing the dressed states and their energies, while non-trivial, is still more efficient than it would have been to diagonalize an $8\times 8$ matrix. Thirdly, as we will explore in greater depth, the coupling matrices for degenerate subspaces are identical, a reflection of the fact that they are indistinguishable under collective operations.

The complexity of computing the eigendecomposition analytically increases rapidly. To our knowledge one can only solve up to $N=8$ spin systems exactly; beyond this point the characteristic polynomial's degree for the largest space is beyond the size where general polynomial solutions exist. Once $N=9$ the largest decomposed matrix will have dimensions $10\times 10$, and since roots always come in positive-negative pairs, the simplified characteristic polynomial will have degree five, which will not generally have a formula for finding the roots.

As we will show in the following section, the problem of diagonalizing the Tavis--Cummings problem is equivalent to diagonalizing a particular two parameter family of Jacobi operators, which we will define as $L(j,k)$. Any real symmetric matrix can be written in a basis where it satisfies the Jacobi operator conditions via a similarity transformation \cite{rutishauser1966jacobi}. Then, if a Jacobi operator was generally solvable in a closed form, all real symmetric matrices` characteristic polynomials would also be solvable in a closed form, a contradiction to Galois's insolubility of general polynomials of degree 5 or greater. 

There is a deep connection of Jacobi operators to the study of orthogonal polynomials, which can in part be seen by the determinant recurrence formula of equation \eqref{eq.det} \cite{teschl2000jacobi}. It is outside the scope of this work to attempt a study of the generated orthogonal polynomials of the Jacobi operators representing the Tavis--Cummings Hamiltonian. As of yet we have been unable to solve the recurrence relationship for the eigenvalues in a closed form. That being said, we suspect that a proof (or disproof) for the existence of a closed form diagonalization of the TC Hamiltonian will be found not with the tools of linear algebra, but with polynomial techniques.
\section{Structure and Statistics of the Full Hamiltonian} \label{sec:results}

Here, we illustrate that a single subspace approximation of the Hamiltonian is generally insufficient to capture the full dynamics of the TC Hamiltonian, regardless of the chosen value of $j$, and provide a more accurate technique for analyzing the TC Hamiltonian. To do so, we first investigate the degeneracy of angular momentum subspaces as a function of $j$, with a focus on determining the maximally degenerate subspace. We then turn our attention to extracting as much information as possible from the collective Lamb shift coupling matrices without numerically solving an eigenvalue problem. As we will show, determining the descriptive statistics of the energy shifts of a given $(j,k)$ subspace is tractable theoretically. Appealing to computational mathematics, we are then able to join the two discussions in order to provide a picture of the degeneracy averaged collective Lamb shift as a function of $N$ and $k$, across all subspaces. Finally, motivated by the numerical results, we determine the root mean square Lamb shift averaged over the degeneracies across all angular momentum subspaces.

\subsection{Maximally Degenerate Angular Momentum Space}

The Dicke subspace is often considered ``special'', in that the following properties hold: all contained states are completely symmetric under particle exchange, the subspace has the largest dimension for a given $N$, the subspace contains the ground state of the Hamiltonian, and the subspace has no degenerate copies. Concerns about the validity of restricting the dynamics to within the Dicke subspace have been noted \cite{baragiola2010collective, wesenberg2002mixed} and we expand on that work here.

The maximally degenerate collective angular momentum subspace for $N$ spin--1/2 particles, which we denote as  $j^*$, is given by:
\begin{equation}\label{jstar}
    j^*=\frac{\sqrt{N}}{2}-\frac{1}{2}+\frac{1}{6\sqrt{N}}+O(N^{-1}).
\end{equation}
The maximally degenerate space is increasingly separated from the Dicke space as $N$ increases. Notice also that the value of $j^*$ does not approach $0$ or $1/2$, indicating that large $N$ structure, through the lens of degeneracy, is not well approximated by a single spin with angular momentum $j=N/2$, nor one with small angular momentum, such as $j=1/2$.

Given that the maximally degenerate angular momentum subspace is well approximated by this expression for $j^*$, we must determine how well this subspace represents the entire system. To formalize this notion, consider $d_{j^*}$, the degeneracy of the $j^*$ subspace, and $d_{j^*+1}$, the degeneracy of the $j^*+1$ subspace. Then, we have that
\begin{equation}\label{adjratio}
    \frac{d_{j^*}}{d_{j^*+1}}=1+O(N^{-3/2}),
\end{equation}
indicating the maximally degenerate subspace is not significantly more degenerate than its nearest neighbor. This argument may be extended to show that generally there's not a large difference in the ratio of nearby spaces. This implies that, although $j^*$ is the most degenerate subspace, we cannot reasonably approximate the system by just this angular momentum space. Within the context of degeneracy--weighted observables, of which an observable for thermal states of the TC Hamiltonian would be, there is no single subspace which can accurately mimic the structure of the entire Hamiltonian.

 To represent a large majority of the possible angular momentum states, we must also include many neighbors of $j^*$ in our analysis. The most essential collection of angular momentum subspaces of the TC Hamiltonian for a given $N$ is well quantified by the strong support of $d_j$. This is visualized in figure \ref{fig:degeneracy_1000}.

\begin{figure}[t]
    \centering
    \includegraphics[scale=0.4]{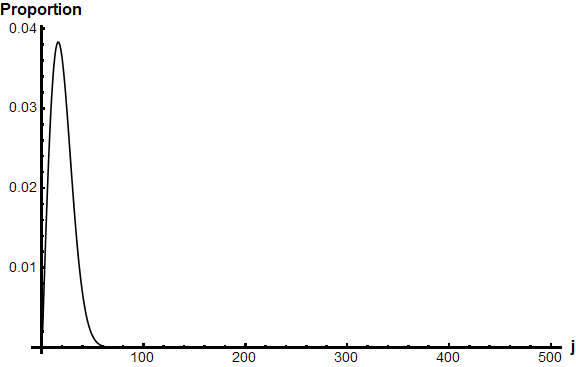}
    \caption{Normalized plot of $d_j$ as a function of $j$, for $N=1000$ spin-1/2 particles. The maximally degenerate space is the $j=15$ angular momentum space. This plot clearly indicates that when weighted by degeneracy, the Dicke subspace contributes negligibly as compared to lower $j$ angular momentum subspaces.}
    \label{fig:degeneracy_1000}
\end{figure}

The strong support of $d_j$ is approximately given by the interval $0 \leq j \leq O(\sqrt{N})$, for all allowed values of $j$. That is, almost all of the states are contained in the subspaces below some constant multiple of $\sqrt{N}$, where the constant is, of course, independent of $N$. The $O(\sqrt{N})$ upper limit can be derived by considering the ratios of the degeneracies of increasingly separated angular momentum subspaces (see appendix for further details). 

%Now, if we try to retain the essence of Holstein--Primakoff approximation in a future analysis by approximating each of the subspaces within the strong support of $d_j$ by non-interacting harmonic oscillators, we will break an additional assumption of the approximation. Notice that the strong support of $d_j$ is defined by the region $0 \leq j \leq O(\sqrt{N})$, meaning that the dimensions of these angular momentum subspaces varies from 1 to $O(\sqrt{N})$. One of the underlying assumptions made in order for the Holstein--Primakoff approximation of the spin-ensemble to be valid is that ``$N$ is large'', which in fact is a condition that the dimension of the Dicke subspace is large. The harmonic oscillator transformation of the Dicke subspace become invalid at higher energies, as a spin-ensemble saturates at a maximum energy of $N/2$ excitations, while a harmonic oscillator allows for infinite excitations. This mismatch is increasingly noticeable as the Dicke subspace shrinks with decreasing $N$. It is then equally invalid to apply the harmonic oscillator transformation central to the Holstein--Primakoff approximation to the subspaces within the strong support of $d_j$, due to the fact that this region contains subspaces with small dimensions.

The insights provided by the computation of $j^*$ and determination of the strong support of $d_j$ have a few important implications. Firstly, the fact that the system must be represented by $O(\sqrt{N})$ subspaces is of interest to those working in the area of the complexity of quantum systems. Secondly, this also will be of interest to those simulating quantum systems admitting a similar angular momentum subspace decomposition, in that so long as one has the subspace structure being preserved, any observable that grows sub--exponentially in $j$, as suggested by (\ref{expsupp}), can be sufficiently modelled using this region of strong support. By restricting computations to this region, we can expect a halving of the dominant order of the computational cost (i.e. an $O(N^4)$ algorithm can be well approximated by an $O(N^2)$ algorithm). In fact, this reduction of order further reduces the effective spin dimensionality of the problem from $O(N^2)$ spin states, to $O(N)$ spin sates, a reduction from the original dimension of $2^N$.

As of this point, we set aside this result and move on to discussing some of the properties that can be gleaned from the coupling matrices as functions of $(j,k)$. In later sections, we average these results across all angular momentum subspaces, using the knowledge gained from this section, to provide an aggregate picture of the energy level structure as a function of $k$ excitations.

\subsection{Statistics of a Collective Angular Momentum Subspace}

In light of the subspace decomposition of the TC Hamiltonian,
\begin{equation}
    \H \cong \bigoplus_{j,k}\big( \omega_0 k \openone_{j,k} +  g_0 L(j,k)\big),
\end{equation}
it is clear that if one were able to diagonalize $L(j,k)$, then the Hamiltonian would be fully solved. It is instructive to visualize the representation of $L(j,k)$ with respect to $\mathcal{B}_{j,k}$:
\begin{equation}
   \begin{bmatrix}
    0 & l_1(j,k) &   & &  \\
    l_1(j,k)  & 0       & l_2(j,k)   &   &   \\
      &  l_2(j,k)  & 0      &  \ddots &   \\
      &         &     \ddots      & \ddots  &   l_{n}(j,k)\\
      &         &           & l_{n}(j,k)   & 0
    \end{bmatrix}.
\end{equation}
Thus, the Lamb shift coupling matrix can be naturally represented as a hollow tridiagonal matrix, a highly structured sparse matrix. Further, this coupling matrix is analogous to an un--normalized transition matrix for a 1D random walk.

A full closed form diagonalization of this matrix is unlikely to exist, but there is still a good amount of information that can be extracted. As a first approach we can consider the problem from a numerical linear algebra perspective. It was shown in 2013 that the eigenvalues of this variety of matrix can be computed exactly (to within numerical precision) in $O(n_{j,k} \log n_{j,k})$ floating point operations, a speed-up over the unstructured problem \cite{coakley2013fast}. This algorithm can then be used to efficiently extract the Lamb shifts of a given $(j,k)$ space on demand, if desired. The eigenvector problem given an eigenvalue $\lambda$ is then solvable in $O(n_{j,k})$ floating point operations utilizing the Thomas algorithm \cite{thomas1949elliptic}. This must be done for each of the $n_{j,k}+1$ eigenvalues, and so while the cost of producing the set of eigenvalues is $O(n_{j,k} \log n_{j,k})$, the cost of producing the entire eigensystem is $O(n_{j,k}^2)$, where the cost is dominated by the eigenvector problem. We expect the numerical speedup of finding the eigenvalues to be useful for simulating this system's dynamics and computing state dependent quantities for states defined by classical mixtures across angular momentum subspace, thereby increasing the maximal value for $N$ that can be feasibly simulated on a classical processor.

We return now to considering properties we can analytically compute, or estimate, of the collection of the Lamb shifts. From work in theoretical numerics, it was shown that the eigenvalues of hollow tridiagonal matrices come in oppositely signed pairs \cite{watkins2005product}. That is, if $\lambda$ is an eigenvalue of $L(j,k)$, then so is $-\lambda$. The eigenvalue spectrum of the Lamb shift coupling matrix can be seen as a two parameter family of sets,
\begin{equation}
    \Lambda(j,k) = \lbrace \lambda \, | \, L(j,k) \bm{v} = \lambda \bm{v}, \bm{v} \neq \bm{0} \rbrace,
\end{equation}
and so from the work of \cite{watkins2005product}, $\Lambda(j,k)$ is an even set. The fact that the eigenvalues of these coupling matrices is a family of sets and not multi-sets is shown in \cite{barth1967calculation}. Thus, if $\abs{\Lambda(j,k)} = \abs{\mathcal{B}_{j,k}} = n_{j,k} + 1$ is odd, there must be exactly one eigenvalue with value $\lambda = 0$. 

A standard parameter of matrices to compute is the determinant, as this is equal to the product of the eigenvalues. There is a two step recursive formula for computing the determinant of a symmetric tridiagonal matrix, that can be used to compute the characteristic polynomial or simply compute the determinant of $L(j,k)$. Let $A$ be a symmetric tridiagonal (Jacobi) matrix with matrix elements,
\begin{align}
    \nonumber A = \sum_{\alpha = 1}^{n+1} a_{\alpha} \ketbra{\alpha}{\alpha} + \sum_{\alpha = 1}^{n} b_{\alpha}\big(&\ketbra{\alpha}{\alpha+1} \\
        &  + \ketbra{\alpha+1}{\alpha}\big),
\end{align}
and sub-matrices $A_\alpha$, formed by discarding all basis vectors with index greater than $\alpha$. Then,
\begin{align} \label{eq.det}
    \nonumber \det(A) &= \det(A_{n+2})\\
        &= a_{n+1} \det(A_{n+1}) - b_{n}^2 \det(A_{n}).
\end{align}

\begin{widetext}

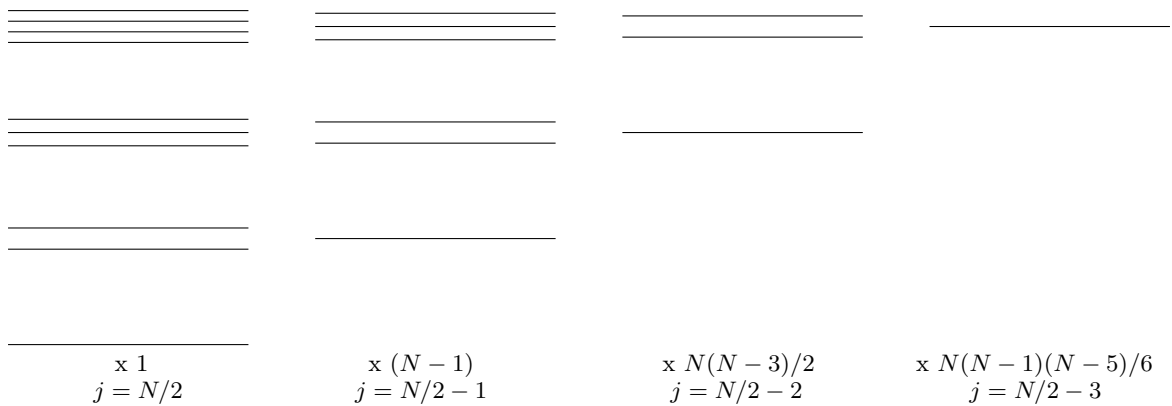
\begin{figure}[h]
    \centering
    \begin{picture}(435,150)
    
\put(32,0){$j=N/2$}
\put(130,0){$j=N/2-1$}
\put(248,0){$j=N/2-2$}
\put(360,0){$j=N/2-3$}

\put(40,10){x 1}
\put(135,10){x $(N-1)$}
\put(245,10){x $N(N-3)/2$}
\put(340,10){x $N(N-1)(N-5)/6$}

\put(0,20){\line(1,0){90}}

\put(0,56){\line(1,0){90}}
\put(0,64){\line(1,0){90}}

\put(0,95){\line(1,0){90}}
\put(0,100){\line(1,0){90}}
\put(0,105){\line(1,0){90}}

\put(0,134){\line(1,0){90}}
\put(0,138){\line(1,0){90}}
\put(0,142){\line(1,0){90}}
\put(0,146){\line(1,0){90}}

\put(115,60){\line(1,0){90}}

\put(115,96){\line(1,0){90}}
\put(115,104){\line(1,0){90}}

\put(115,135){\line(1,0){90}}
\put(115,140){\line(1,0){90}}
\put(115,145){\line(1,0){90}}

\put(230,100){\line(1,0){90}}

\put(230,136){\line(1,0){90}}
\put(230,144){\line(1,0){90}}

\put(345, 140){\line(1,0){90}}

\end{picture}

    \caption{Schematic representation of the energy eigenstates of the Tavis-Cummings Hamiltonian with excitations $0 \leq k \leq 3$ along the vertical, and labelled horizontally by the number of degeneracies of each angular momentum subspace.}
    \label{fig:TC_ladder}
\end{figure}

\end{widetext}

Upon computing the determinant of $L(j,k)$, we find that if $n_{j,k}+1$ is odd, then the recurrence terminates with $\det A_0 = 0$, and so $\det L(j,k) = 0$. Otherwise, $n_{j,k}+1$ is even and the determinant is given as
\begin{equation}
    \det L(j,k) =(-1)^{\frac{n+1}{2}} l_{n}^2 l_{n-2}^2 \cdots l_1^2,
\end{equation}
where the dependence of the matrix elements $l_\alpha$ on $(j,k)$ were suppressed for clarity.

While it is interesting to know that the determinant can be computed efficiently and in a closed form, it does not provide a description of the structure of the collective Lamb shifts. Rather, given the set of Lamb shift eigenvalues, $\Lambda(j,k)$, it is more instructive to provide descriptive statistics. Given knowledge of the eigenvalues, the $t$-th moment is given as
\begin{equation}
    \avg{\Lambda(j,k)^t} = \frac{1}{\abs{\mathcal{B}_{j,k}}}\sum_{\lambda\in\Lambda(j,k)} \lambda^t.
\end{equation}
We can avoid computing the eigenvalues explicitly by noticing that the sum over eigenvalues is equivalent to the trace of the Lamb shift coupling matrix. Thus,
\begin{equation}
    \avg{\Lambda(j,k)^t} = \frac{1}{\abs{\mathcal{B}_{j,k}}} \Tr\big(L(j,k)^t\big).
\end{equation}

Then, it is clear that the mean of each subspace is 0, which follows as the Lamb shift coupling matrix is hollow,
    \begin{equation}
        \avg{\Lambda(j,k)} = 0,
    \end{equation}
as all the diagonal entries of $L(j,k)$ are zero with sum zero. This statement can be extended to all odd moments of the Lamb shift eigenvalues. That is, for each coupling matrix, $L(j,k)$, 
\begin{equation}
    \avg{\Lambda(j,k)^{2t+1}}=0,\,\,\forall t \in \mathbb{N}.
\end{equation}
This follows immediately from the fact that, for every $\lambda\in\Lambda(j,k)$, $-\lambda\in\Lambda(j,k)$.

In order to quantify the magnitude of the collective Lamb shift splittings, we may utilize the variance as a measure, which in this case is equal to the second moment of $\Lambda(j,k)$:
\begin{align}
    \nonumber \operatorname{Var}(\Lambda(j,k)) &= \avg{\Lambda(j,k)^2} - \avg{\Lambda(j,k)}^2\\
    &= \avg{\Lambda(j,k)^2}.
\end{align}
Computing the variance is then equivalent to determining the trace of the square of the coupling matrix, which is a banded pentadiagonal matrix, explicitly given as
\begin{equation}
    \begin{bmatrix} 
    l_1^2 & 0 & l_1 l_2  &   &  & \\
    0 & l_1^2 + l_2^2 & 0 & l_2 l_3 & & \\
    l_1 l_2 & 0 & \ddots & \ddots & \ddots &\\
     &  l_2 l_3 & \ddots & \ddots & \ddots & l_{n-1}l_{n}&\\
     & & \ddots & \ddots & l_{n-1}^2 + l_{n}^2&0 &\\
     & & & l_{n-1}l_{n} & 0 & l_{n}^2
    \end{bmatrix}
\end{equation}

Thus, the trace of the square of $L(j,k)$ has a tidy closed form expression in terms of the matrix elements $l_\alpha(j,k)$,
\begin{equation}
    \Tr L(j,k)^2 = 2\sum_{\alpha=1}^{n} l_\alpha (j,k)^2.
\end{equation}
Then, the variance of $\Lambda(j,k)$, for $k \geq k_0(j)$, with $k' = k-k_0(j)$, is given by the following expression
\begin{align}\label{eq:variance}
    \nonumber \operatorname{Var}(\Lambda(j,k)) &= \frac{1}{2}\abs{\mathcal{B}_{j,k}}^3 -     \frac{1}{3}\abs{\mathcal{B}_{j,k}}^2(2k'+4j+7)\\
    \nonumber &+ \abs{\mathcal{B}_{j,k}}(2jk'+2k'+4j+7/2)\\
    &- \frac{1}{3}(6jk' + 8j + 4k' + 5)
\end{align}

 Recalling that the dimension of the basis of a $(j,k)$ space is given by $\abs{\mathcal{B}_{j,k}} = \min \lbrace 2j+1 , k - k_0(j) + 1\rbrace$, when $k$ satisfies $k-k_0(j) > 2j$ the dimension of the space becomes fixed at $2j+1$. And so, for $k$ such that $k-k_0(j) > 2j$, or equivalently such that $k > N/2 + j$, $\operatorname{Var}(\Lambda(j,k))$ is a linear function in $k$. This can be seen by substituting $\abs{\mathcal{B}_{j,k}} = 2j+1$ into equation \eqref{eq:variance}.

Taking the square root of the variance, we can attain the standard deviation, which has an interpretation as the average distance from the mean. In this sense then, the average collective Lamb shift, treating $j$ as a constant, is $O(\sqrt{k})$ for $k>N/2+j$. In order to describe the full statistics of the collective Lamb shift, it is insufficient to consider $j$ a constant, rather we must consider all angular momentum subspaces and their respective degeneracies present at a given value of $k$.

\subsection{Rotating--Wave Approximation Revisited}\label{rwa}

Before we move to averaging over degeneracies, we include one more aspect of the Lamb shifts that may be computed generally and discuss it's implications. Since all $L(j,k)$ are non-negative matrices we can bound the maximal absolute value of the eigenvalue, also given by the spectral norm, from above and below using the Perron--Frobenius theorem:
\begin{equation}
    \min_m \sum_n [L(j,k)]_{mn}\leq \max \Lambda\leq \max_m \sum_n [L(j,k)]_{mn}.
\end{equation}
Applying this theorem allows us to determine that $\max \Lambda(j,k)$ is upper bounded by the various cases illustrated in \eqref{eq:lambda_max_cases}.
\begin{equation}\label{eq:lambda_max_cases}
    \begin{cases}
        \frac{2}{\sqrt{3}}\sqrt{(2j+k')jk'} & \text{generally}\\
        2[j\sqrt{k'}-\frac{1}{2}\frac{j^2}{\sqrt{k'}}+\frac{1}{8}\frac{j^4}{(k')^{5/2}}+O(\frac{j^5}{(k')^{7/2}})] & 2j\ll k'\\
        2[\frac{1}{\sqrt{2}}k'\sqrt{j}-\frac{1}{8\sqrt{2}}\frac{(k')^2}{\sqrt{j}}+\frac{1}{512}\frac{(k')^4}{j^{5/2}}+O(\frac{(k')^5}{j^{7/2}})] & k'\ll 2j.
    \end{cases}
\end{equation}

The relations for $\max \Lambda(j,k)$ can narrow the energy range we need to consider in experimental design for a given value of $N$ and bounded total energy. We could likewise provide a lower bound on the maximal splitting via the same argument--this always yields the minimum of the first row and the last row (both of which have a single entry in the coupling matrix). 

These bounds on the maximal eigenvalue also have an additional implication. As a general rule of thumb, the rotating--wave approximation used for approximating the Dicke Hamiltonian by the Tavis--Cummings Hamiltonian is said to be valid for $g_0\sqrt{N}\ll \omega_0$. Although this rule of thumb is useful for potentially determining the ability to experimentally resolve vacuum Rabi oscillations of a spin ensemble with a cavity, it is not a good metric for determining the validity of the RWA. We can improve the specificity of this requirement. We begin by noting that from the lower bound for the maximal eigenvalue, we have:
\begin{equation}
    \lim_{j,k\rightarrow\infty} \|L(j,k)\|_\infty=\infty,
\end{equation}
which means that eventually $g_0\max \Lambda(j,k)$ will approach and exceed $2\omega_0$. This means that the true condition that should be used to justify a rotating--wave approximation is 
\begin{equation}
g_0 \max \Lambda(j,k)\ll \omega_0,
\end{equation}
which puts a limit on the size of $j$ and $k$ that can be considered with this model. A maximally allowed value for $k$, after which the rotating wave approximation breaks down, is not a property unique to the TC Hamiltonian, as the JC Hamiltonian's RWA is invalidated when $k \approx \omega_0^2 / g_0^2$. Given our upper and lower bounds on $\max\Lambda(j,k)$ we can estimate where the rotating--wave approximation begins to breakdown. As an example, we consider the behavior of the density of states for an $N=20$ system.

The density of states is a sum of delta functions over all excitation spaces, with the location of the delta functions being the energies of the Lamb-shifted eigenstates, scaled by the weight:
\begin{equation}
    n(E) = \sum_{k=0}^\infty\sum_{\lambda\in\Lambda(k)}w_k(\lambda)\delta(E-(k\omega_0 + \lambda g_0)).
\end{equation}
When the rotating wave approximation holds, the distribution of delta functions across neighboring excitation subspaces will be well separated, as show in figure \ref{fig:DOS_sep}. On the other hand, figure \ref{fig:DOS_overlap} illustrates what the energy level structure looks like when the rotating wave approximation breaks down. In this case, states in a given excitation subspace can overlap with states from neighboring excitation subspaces, breaking the notion of the good quantum number, and invalidating the predictions of the model.

\begin{figure}[t]
    \centering
    \includegraphics[scale=0.4]{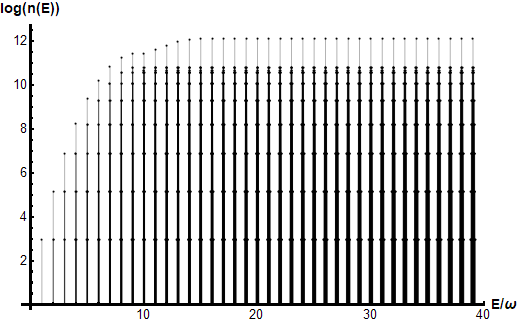}
    \caption{Scaled density of states for $N=20$ spins, with $\omega/g = 500$.}
    \label{fig:DOS_sep}
\end{figure}

\begin{figure}[h!]
    \centering
    \includegraphics[scale=0.4]{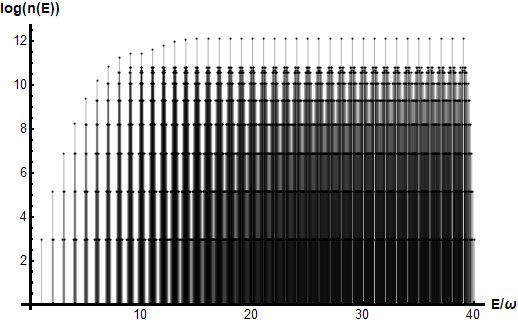}
    \caption{Scaled density of states for $N=20$ spins, with $\omega/g = 100$.}
    \label{fig:DOS_overlap}
\end{figure}

\subsection{Degeneracy Averaged Statistics of the Collective Lamb Shift}

Given that relaxation processes and thermal excitation tend to suppress collective behavior in an ensemble and spread population over many subspaces \cite{wood_cavity_2016,baragiola2010collective,wesenberg2002mixed,chase_collective_2008}, the utility of descriptive statistics of the collective Lamb shifts for specific values of $j$ is limited. To address this constraint, we now discuss properties of the collective Lamb shift upon taking an appropriate average over angular momentum subspaces. To begin, we define a probability distribution on the set of eigenvalues across all values of $j$ at a given value of $k$. A natural choice is weighting each eigenvalue by its degeneracy.
\begin{equation}
    w_k(\lambda) = \sum_{j} \begin{cases}
        d_j & \lambda \in \Lambda(j,k)\\
        0 & \text{else}
    \end{cases}
\end{equation}
The sum over $j$ accounts for the case of repeat eigenvalues across $j$ spaces, although we believe that it is generally only the 0 eigenvalue that repeats across $j$ spaces. For convenience, we also define the set of Lamb shift eigenvalues over $k$ excitations to be given as:
\begin{equation}
    \Lambda(k) = \bigcup_{j} \Lambda(j,k).
\end{equation}
The set of pairs, $(\lambda, \omega_k(\lambda))$ define an unnormalized probability distribution on the Lamb shifts for an $N$ spin TC system with $k$ excitations.

The $t$-th moment of the collective Lamb shifts across all angular momentum spaces is formally written,
\begin{equation}
    \avg{\Lambda(k)^t} = \frac{1}{D_k}\sum_{\lambda\in\Lambda(k)}w_k(\lambda)\lambda^t,
\end{equation}
where $D_k$ is the number of states with $k$ excitations, given by
\begin{equation}
    D_k = \sum_{k' = 0}^k \binom{N}{k'}.
\end{equation}
\newline

Recall that if $k<k_0(j)=N/2-j$, then there are no states of excitation $k$ for the given value of $j$. In this case, $\Lambda(j,k)$ is empty and does not contribute to the statistics of this excitation level.  Once $k \geq N$, the total number of states present at a given excitation becomes fixed at $D_k = 2^N$.

As with the case of a single angular momentum space, it is better to compute the moments utilizing traces of the Lamb shift coupling matrix, which are computable in $O(n_{j,k})$ floating point operations, compared to the $O(n_{j,k}\log n_{j,k})$ eigenvalue problem. The advantage is particularly significant for the second moment, where we already have determined a closed form expression for the trace, dropping the cost to $O(1)$ operations per angular momentum subspace.

Using this insight, the $t$-th moment of the collective Lamb shift splittings across all angular momentum spaces is equal to the weighted average of traces of the coupling matrices:
\begin{equation}\label{degenmoment}
    \avg{\Lambda(k)^t} = \frac{1}{D_k}\sum_{j} d_j \Tr (L(j,k)^t).
\end{equation}
As $\Lambda(k)$ is the union of even sets, it is also an even set. This fact allows us to immediately determine that the odd moments for the collective Lamb shifts indexed by $k$ excitations are all 0:
\begin{equation}
    \avg{\Lambda(k)^{2t+1}}=0,\,\,\forall t \in \mathbb{N}.
\end{equation}

Due to the combinatorial nature of the weights on eigenvalues, there is no exact closed form expression for the even moments of the Lamb shift splittings. That being said, it is computationally feasible to visualize the function for select values of $N$, as shown in figure \ref{fig:variance}. 

\begin{widetext}

\begin{figure}[h]
    \centering
    \includegraphics{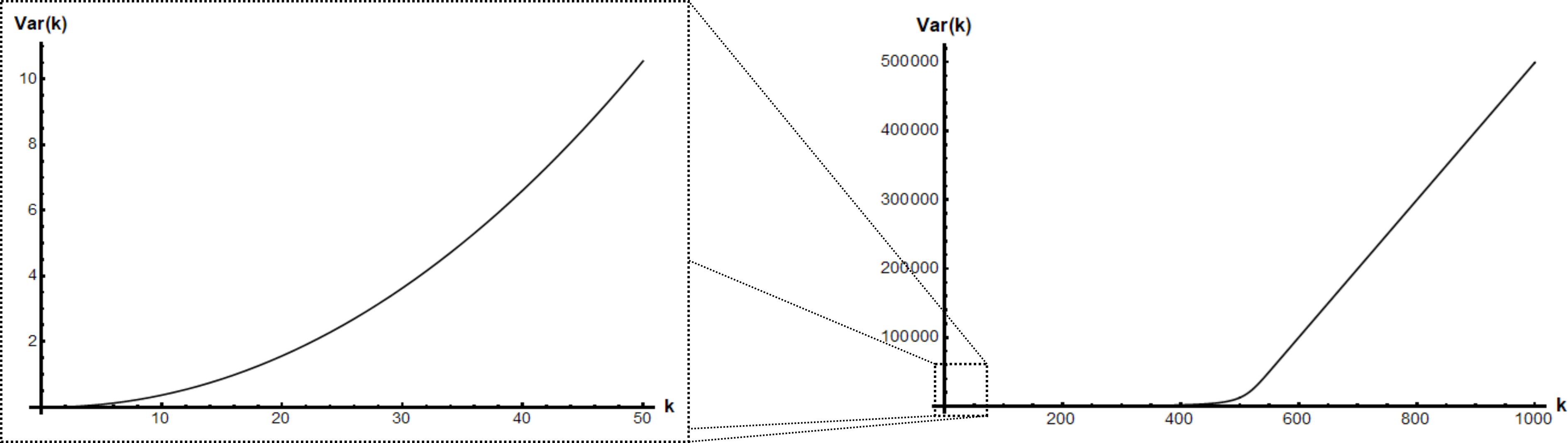}
    \caption{Variance of the unit--less ($\hbar=g_0=1$) Lamb shift splittings for $N=1000$ spin--1/2 particles. Notice that the variance becomes linear in $k$ soon after $k = N/2 = 500$.  Notice the non-linearity and reduction of scale of the variance in the lower excitation subspaces as compared to the $k>N/2$ subspaces.}
    \label{fig:variance}
\end{figure}

\end{widetext}

We first address the apparent suppression of the variance for $k<N/2$. Recalling from the partial energy level diagram of figure \ref{fig:TC_ladder}, as we introduce a new angular momentum subspace into the statistics, the ground state of that $j$ space is included into the statistics with energy splitting of 0. Now, since $d_j$ is an increasing function for $j<j^*$, the dominant element of the distribution of eigenvalues is the ground state of the smallest considered angular momentum subspace. This trend holds true until $k$ approaches $N/2-j^*$. In the case of $N=1000$, we predict $j^* = 15$, hence the suppression of the variance to nearly $k=N/2$.

As for the linearity of the variance starting at roughly $k=N/2$, we expect the variance of each angular momentum $j$ subspace to be linear in $k$ for values of $k>N/2+j$. It follows that the apparent transition into the linear regime at $k\approx N/2$ is caused by the variance of subspaces in the region of strong support of $d_j$ being most dominant in the linear regime for $k > N/2 + j^*$. Thus, one can expect a linear variance in $k$ for values of $k > N/2 + \sqrt{N}/2 - 1/2 + 1/(6\sqrt{N})$, which is dominated by $N/2$ for large $N$.

Now, while we showed that the variance is indeed a linear function in $k$, we did not provide the function, as the slope of this linearity is a sizeable polynomial in $j$. We will return to the task of averaging this polynomial over all possible values of $j$ momentarily. As a first investigation, we can efficiently illustrate the trend numerically. We perform regression on the linear regime of the variance for select values of $N$, and plot the slope of these lines as a function of $N$, as seen in figure \ref{fig:slope_regression}.

\begin{figure}[h!]
    \centering
    \includegraphics[scale=0.4]{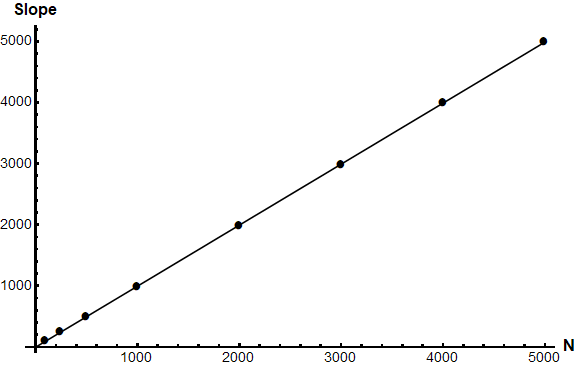}
    \caption{Slope of the variance of the collective Lamb shift splittings in the linear regime, for various $N$. Points mark computed values of the slope for each $N$. The regression model is Slope$(N) = 0.9989 N - 0.27$, with an $R^2$ value of nearly 1.}
    \label{fig:slope_regression}
\end{figure}
  
The result of the regression indicates that for the values of $k$ and $N$ considered, the variance is effectively growing at a rate of $0.9989 N k$. In fact, we are able to show analytically that the variance of the collective Lamb shift splittings, as a function of $k$ excitations, grows as the product of $N$ and $k$ for $k \geq N$, such that
\begin{equation}\label{degenavg}
    \operatorname{Var}(\Lambda(k)) = O(N k).
\end{equation}
We further conjecture that this result holds for all $k > k^*$, where $k^*$ is some value in the range $N/2 < k^* < N$, and likely depends on $N$. This can be seen in part by noting that the dimension of a subspace, $\abs{\mathcal{B}_{j,k}}$, saturates at $2j+1$ when $k = k_0(j) + 2j = N/2 + j$. Further, since we need only consider the strong support of $d_j$ when averaging (for more details see the appendix), the result could likely be extended to hold if $k^* = N/2 + O(\sqrt{N})$. In the region $k^* < k < N$, the statistics are not exposed to all possible spin states as $D_k < 2^N$. This issue can be avoided by bounding $d_j / D_k$ instead of $d_j/2^N$, for all values of $k$ greater than some $k^*$. Since $D_k$ is almost constant on this region, the extension should not be too difficult. The proof of equation \eqref{degenavg} is given in the appendix, for the case where $k \geq N$. 

%Given that the Jaynes-Cummings Hamiltonian has an energy level splitting of magnitude $g_0 \sqrt{k}$, it is perhaps surprising that the average energy level splitting of the Tavis-Cummings Hamiltonian, in the RMS sense, has a magnitude approximately given by $g_0 \sqrt{N k}$. Further, the proof of equation \eqref{degenavg} illustrates that, given enough excitations, the structure is dominated by the subspaces near the maximally degenerate angular momentum subspace. In fact, only the lowest $O(\sqrt{N})$ subspaces were theoretically needed to produce the result. If one took the Holstein-Primakoff approximation and attempted to predict this structure, they would be missing out on the most dominantly contributing factors; the Dicke subspace does not contribute appreciably to this result. Of course, the larger angular momentum spaces are more important when one restricts to smaller excitation levels, or equivalently, to very low temperature samples.% This then brings into question the utility of the Holstein-Primakoff approximation entirely, as we have shown that restriction to the lowest 9 excitation subspace still allow for exact diagonalization for arbitrary $N$. We will explore the exact temperature dependent conditions for a true low-subspace restriction in a future work.
\section{Discussion} \label{sec:discussion}
    It is a common practice to treat the Tavis--Cummings Hamiltonian as a generalized Jaynes--Cummings Hamiltonian operating in the Dicke subspace ($j=N/2$) with a collectively enhanced spin-cavity coupling, $g_{eff} = g_0 \sqrt{N}$. The Lamb shifts in higher excitation manifolds, then, scale in analogy to the Jaynes--Cummings model, $g_0\sqrt{k}\rightarrow g_{eff}\sqrt{k} = g_0\sqrt{Nk}$. The justification for this approximation is often taken to be operation at low excitation, with the limiting case being the single-excitation manifold ($k=1$) where the Lamb shift is simply $g_0\sqrt{N}$. It is perhaps surprising, then, that the average Lamb-shifted energy level splitting of the Tavis--Cummings Hamiltonian, taken over all angular momentum subspaces, has a magnitude approximately given by $g_0 \sqrt{N k}$. In this sense, the dynamics of a large ensemble appear similar for very low excitations ($k \sim 1$) and very high excitations ($k \gg N/2$). Mathematically, this is due to the structure being dominated by the subspaces nearest to the maximally degenerate subspace, as given in the proof of equation \eqref{degenavg}. In fact, only the lowest $O(\sqrt{N})$ subspaces are theoretically needed to produce the result. 
    
    We note the consistency of our result with various experiments measuring a high--cooperativity splitting of spin ensembles interacting with high Q cavities, where a $g_0 \sqrt{N}$ behavior is observed \cite{rose_coherent_2017,angerer_superradiant_2018,kubo_strong_2010,schuster_high-cooperativity_2010}. These experiments are generally run at relatively high power, corresponding to many excitations in the system. Further, it has been noted that the high--cooperativity splitting disappears at high drive powers \cite{angerer_ultralong_2017,chiorescu2010magnetic}. Our results indicate that the coalescence of the splitting into a single peak may be considered as a violation of the rotating-wave approximation (Section \ref{rwa}), such that the eigenstructure of the Tavis--Cummings model becomes invalid, and ``classical''
    behavior emerges due to the smearing of the density of states (Figure \ref{fig:DOS_overlap}). 
    
    We also note that care must be taken in applying this result, however, as the linearity of the variance in $k$ is invalid for $k<N/2$, as illustrated in figure \ref{fig:variance}. This indicates that, while the single excitation splitting prediction of $O(\sqrt{N})$ is indeed valid for states with enough energy under the right conditions, for non--trivial moderate excitation states the variance cannot be approximated as simply. The full implications of this result are outside the scope of this paper, but we note our results suggest a further examination of the validity of various subspace restriction techniques is required in the regime of low to moderate excitation and ensemble size used for many quantum devices.

\section{Conclusion} \label{sec:conclusion}
   
    In this work we have revisited the Tavis--Cummings model and have elucidated a number of new observations. We began by recasting the original decomposition of the Hamiltonian for this system as a direct sum of subspaces, with respective structure described by a two-parameter family of coupling matrices. We further analyzed the structure of the degeneracies within the Hamiltonian and found the system is well described by a (relatively) small subset of the possible angular momentum spaces. This identification alone has implications to fields such as the complexity and simulation of quantum systems, as well as direct applications to the latter parts of our work.
    
    We proceeded to describe various parameters of our coupling matrices that we can compute and what those computed values could imply. First, we showed that, due to the structure of these coupling matrices, finding the eigenvalues is computationally easier than for general matrices. Additionally, we showed that all odd moments of the distribution of the eigenvalues are zero and computed the second moment explicitly. Using the derived formulae of this paper, one could also compute the higher order even moments, both within each subspace but also averaging across all degeneracies. The utility of such computations is dependent on the feasibility of experimentally detecting these higher moments. From this we provided bounds on the maximal value in the Lamb shift collection for each $(j,k)$ subspace. The bounds provided can likely be tightened; however, the current bounds still provide a strong condition for determining the validity of the Tavis--Cummings rotating--wave approximation. 
    
    %Although at higher excitation levels the rotating-wave approximation breaks down, one should consider how much energy is required in the system for those states to become populated. In particular, in light of thermal laboratory conditions, the temperature of the surrounding environment can be used to determine which states are appreciably populated, thus bounding the allowed value for $k$. Further, it's also possible that $j$ can become too large for the rotating-wave approximation to hold. This physically corresponds to the number of particles in our system being too numerous, and so the maximal splitting approaches $\omega_0/g_0$.
    
    Following this, we showed how one would average over the $j$ portion of the computed statistics of the Lamb shifts on the coupling matrices by averaging over the angular momentum, using degeneracy as the weighting function. We bounded the RMS of the splittings for $k\geq N$, and showed good agreement with numerics. A natural question is whether this can be shown for some $k^*$ where $N/2<k^*<N$, or even for lower values of $k$. We believe that such a $k^*$ exists, and provided a sketch of what a proof for this extension might look like.
    
    %We then finished off by providing a discussion on possible measures of collectiveness for ensemble systems which are similar in nature to ours. We do not determine any level of correctness of these measures, but merely propose them. A precise definition for collectiveness is hard to pin down, but feels like a natural term to use in many situations, especially while describing states, operators, and evolution.
    
    %For possible extensions of our work, we imagine not only averaging over angular momentum spaces, weighting by degeneracies, but also averaging over excitation spaces, weighting by a thermal distribution. Additionally, we have only described statistics and properties of the Lamb shift splittings, however, a great extension would be to consider similar properties under a given physical observable, and compare this to experimental observations for these systems.
\section*{Acknowledgements} \label{sec:acknowledgements}

We thank Maryam Mirkamali for helpful discussions on mesoscopic physics and multi--body entanglement, and John Watrous for inspiring us to reexamine this problem using symmetry techniques, which gave us the tools to better understand the structure of the Tavis--Cummings model.

\section*{Funding}
This work was supported by Industry Canada, the Canada
First Research Excellence Fund (CFREF), the Canadian Excellence Research Chairs (CERC 215284) program, the Natural Sciences and Engineering Research Council of Canada
(NSERC RGPIN-418579) Discovery program, the Canadian
Institute for Advanced Research (CIFAR), and the Province
of Ontario.

\bibliographystyle{unsrt}
\bibliography{refs}
%\nocite{*} this makes works appear even if they aren't cited.
\clearpage

\appendix*
\begin{widetext}
\section{Appendix} \label{sec:appendix}

\begin{proof}[Proof of equation (\ref{jstar})]
To derive the value of $j^*$, we begin by defining the degeneracy function in a convenient form,
\begin{equation}\label{eq.choose}
    f(j)=\frac{2j+1}{N/2+j+1} {N \choose N/2+j},
\end{equation}
where $j$ takes integer or half-integer values $0 \leq j \leq N/2$, depending on the parity of $N$. To prepare for differentiation, the binomial coefficient can be extended to a continuous function 
\begin{equation}
    {N \choose K} = \frac{\Gamma(N+1)}{\Gamma(K+1)\Gamma(N-K+1)},
\end{equation}
such that \eqref{eq.choose} can be written as a continuous function in $j$
\begin{equation}
    f(j) = \frac{2j+1}{N/2+j+1}  \frac{\Gamma(N+1)}{\Gamma(N/2 + j+1)\Gamma(N/2 - j + 1)}.
\end{equation}

We may now differentiate and look for critical values:

\begin{equation}\label{eq:diff_degen}
    \frac{d}{dj}f(j)=4{N\choose N/2+j}\frac{\frac{1}{2}(2j+1)(N+2j+2)(H_{N/2-j}-H_{N/2+j})+N+1}{(N+2j+2)^2}.
\end{equation}

In equation \eqref{eq:diff_degen}, $H_x$ is the Harmonic series truncated at term $x$. The degeneracy is then maximal when
\begin{equation}
    \frac{1}{2}(2j+1)(N+2j+2)(H_{N/2-j}-H_{N/2+j})+N+1=0.
\end{equation}
We can re-cast this result by utilizing the expression for $H_x=\log x+\gamma+O(x^{-1})$, where $\gamma$ is the Euler-Mascheroni constant ($\gamma\approx 0.577$). This allows us to take the difference of the Harmonic numbers as $\log$'s, which cancels the additive term. After simplifying, we are left with
\begin{equation}
    \frac{1}{2}(2j+1)(N+2j+2)\log(\frac{N/2-j}{N/2+j})+N+1=0.
\end{equation}

This is a very tight approximation. If we examine the series expansions of this, we see that taking $j\approx \frac{\sqrt{N}}{2}$ will remove the leading error. We may repeat this procedure, noting that the errors are a Laurent series in $\sqrt{N}$, so we adjust by decreasing powers of $\sqrt{N}$ corrections. Using this, we take as a guess that $j=\frac{\sqrt{N}-1}{2}+\frac{1}{6\sqrt{N}}$, which yields:
\begin{equation}\label{eq:plug_in_guess}
    (3N+1)(3(N^{3/2}+N+\sqrt{N})+1)\frac{\log(\frac{-6N+6\sqrt{N}-2}{3\sqrt{N}(N+\sqrt{N}-1)+1}+1)}{18N}+N+1.
\end{equation}
When expanded as a series in the limit of large $N$, we have:
\begin{equation}
    \frac{1}{\sqrt{N}}+O(\frac{1}{N}),
\end{equation}
and thus our guess is equal to the true root in the limit of $N\longrightarrow\infty$. Thus,
\begin{equation}
    j^*= \frac{\sqrt{N}-1}{2}+\frac{1}{6\sqrt{N}} + O(N^{-1})
\end{equation}
is the collective spin space with the largest degeneracy, up to error $O(N^{-1})$.
\end{proof}

\begin{proof}[Proof of equation (\ref{adjratio})]
We begin by considering the ratio:
\begin{eqnarray}
    \frac{d_{j^*}}{d_{j^*+1}}&=& \frac{N! (2j^*+1)}{(N/2 - j^*)!(N/2+j^*+1)!}\cdot \frac{(N/2 - j^*-1)!(N/2+j^*+2)!}{N! (2j^*+3)}\\
    &=& \frac{2j^*+1}{2j^*+3}\cdot \frac{N/2+j^*+2}{N/2-j^*}\\
    &=& (1-\frac{2}{2j^*+3})\cdot \frac{1+\frac{2j^*}{N}+\frac{4}{N}}{1-\frac{2j^*}{N}}\\
    &=& (1-\frac{2}{2j^*+3})(1+\frac{2j^*}{N}+\frac{4}{N})(1+\frac{2j^*}{N}+(\frac{2j^*}{N})^2+O(((j^*)/N)^3))
\end{eqnarray}
where the last factor is a geometric series expansion with a ratio of $\frac{2j^*}{N}$. Grouping by powers, using $j^*=O(\sqrt{N})$, we have:
\begin{equation}
    1+[-\frac{2}{2j^*+3}+\frac{2j^*}{N}+\frac{2j^*}{N}]+[\frac{4}{N}-2\frac{4j^*}{(2j^*+3)N}+2(\frac{2j^*}{N})^2]+O(N^{-3/2})
\end{equation}
Now we utilize $j^*=\frac{\sqrt{N}}{2}-\frac{1}{2}+\frac{1}{6\sqrt{N}}$ to evaluate the above:
\begin{eqnarray}
    &=& 1+0+[-\frac{2}{N}+\frac{4}{N}-\frac{4\sqrt{N}}{\sqrt{N}N}+2(\frac{1}{\sqrt{N}})^2]+O(N^{-3/2})
\end{eqnarray}
So the ratio of the degeneracies is $1+O(N^{-3/2})$.
\end{proof}

\begin{proof}[Proof of Strong Support for $0\leq j\leq O(\sqrt{N})$]

Recall the computation of the relative population between the maximal angular momentum subspace and its neighbor, and consider the case when the leading term will contribute to the ratio of neighboring values of $j$. We found that when $\frac{2j}{N}\ll 1$,
\begin{equation}
    \frac{d_j}{d_{j+1}}=1+[\frac{4j}{N}-\frac{2}{2j+3}]+O(N^{-1}),
\end{equation}
which cancelled when $j=j^*$ since $j^*=\frac{\sqrt{N}}{2}+O(1)$. 

Suppose we still have $\frac{2j}{N}\ll 1$, but now we consider a subspace nearby the maximal angular momentum space, such that $j=j^*+\Omega(\sqrt{N})$. The ratio for this value of $j$ is then $1+\Omega(N^{-1/2})$. This ratio will remain valid for increasing values of $j$, so long as $\frac{2j}{N}\ll 1$.

To find the ratio of the degeneracies of the next nearest neighbors, we apply this procedure twice, finding 
\begin{equation}\label{expsupp}
    \frac{d_j}{d_{j+2}}=(1+\Omega(1/\sqrt{N}))^2.
\end{equation}
To continue this argument to further subspace degeneracies, listed within the set $\{d_j,d_{j+1},\ldots, d_{j+O(\sqrt{N})}\}$, we can write a geometric series taken at the infinity limit since the other term will contribute a much smaller portion to the sum. Thus, this series limits to
\begin{equation}
    \frac{1}{1-\frac{1}{1+\Omega(1/\sqrt{N})}}=O(\sqrt{N}).
\end{equation}
We have shown that while the total number of allowed values for $j$ is $O(N)$, the fractional contribution contained in this region is only $O(N^{-1/2})$. Thus we have that of the $2^N$ possible angular momentum states, \textit{most} of them are contained within the lowest, smallest values of $j$, $O(\sqrt{N})$ angular momentum subspaces.
\end{proof}

\begin{proof}[Proof of equation \eqref{eq:variance}]
We will make use of the following summation formulae,
\begin{align*}
    \sum_{i=1}^{n} i &= \frac{n(n+1)}{2}\\
    \sum_{i=1}^{n} i^2 &= \frac{n(n+1)(2n+1)}{6}\\
    \sum_{i=1}^{n} i^3 &= \frac{n^2(n+1)^2}{4}.
\end{align*}
Recall that the trace of the square of the coupling matrix can be written exactly as
\begin{equation}
     \Tr L(j,k)^2 = 2 \sum_{\alpha=1}^{n} l_\alpha(j,k)^2 = 2 \sum_{\alpha=1}^{n} \big(2\alpha j-\alpha(\alpha-1)\big)\big(k'-\alpha+1\big).
\end{equation}
Now, grouping the summand by orders of $\alpha$ we have
\begin{equation}
    l_\alpha^2(j,k) = \alpha^3 - \alpha^2(2j+1+k'+1) +\alpha(2j+1)(k'+1).
\end{equation}
And so
\begin{align}
    \nonumber \Tr L(j,k)^2 &= \frac{1}{2}\abs{\mathcal{B}_{j,k}}^4 - \frac{2}{3}\abs{\mathcal{B}_{j,k}}^3(2j+k'+7/2) + \abs{\mathcal{B}_{j,k}}^2(4j+2k'+2jk'+7/2)\\
    &- \frac{2}{3}\abs{\mathcal{B}_{j,k}}(3jk' + 4j + 2k + 5/2).
\end{align}
The variance is then, by definition,
\begin{align}
    \nonumber \operatorname{Var}(\Lambda(j,k)) &= \frac{1}{2}\abs{\mathcal{B}_{j,k}}^3 -\frac{1}{3}\abs{\mathcal{B}_{j,k}}^2(2k'+4j+7)\\
    &+ \abs{\mathcal{B}_{j,k}}(2jk'+2k'+4j+7/2) - \frac{1}{3}(6jk' + 8j + 4k' + 5).
\end{align}
\end{proof}

\begin{proof}[Proof of equation \eqref{eq:lambda_max_cases}]
We begin by transforming the entry values into a continuous function of $\alpha$ so that we may differentiate it. We will use Perron-Frobenius since we have a non-negative matrix and so can bound the maximal eigenvalue by the maximal row sum. To this end we focus on maximizing a single $l_\alpha(j,k)$ entry and double it since the true maximum will occur within one entry of the optimal continuous value choice for $\alpha$ and there are two entries in that row.

Differentiating this we have:
\begin{equation}
    \frac{d}{d\alpha}\left[ \sqrt{\alpha}\sqrt{2j+1-\alpha}\sqrt{k'-\alpha+1}\right] =\frac{3\alpha^2-4\alpha+2j(-2\alpha+k'+1)-2\alpha k'+k'+1}{2\sqrt{\alpha(\alpha-2j-1)(\alpha-k'-1)}}
\end{equation}
and so this is optimized when:
\begin{eqnarray}
    \alpha&=&\frac{1}{3}\left(2j+k'+2\pm \sqrt{4j^2-2jk'+2j+(k')^2+k'+1}\right)\\
    &=& \frac{1}{3}\left(2j+k'+2\pm \sqrt{(2j+\frac{1}{2})^2+(k'+\frac{1}{2})^2-2jk'+\frac{1}{2}}\right)
\end{eqnarray}
In the above we must exclude the positive sign choice since this results in $\alpha\geq \abs{\mathcal{B}_{j,k}}$, which is beyond the domain for $\alpha$. With the negative sign choice we note that $\alpha$ is linear in $j$ and $k'$ to first order, so we remove the 1 shifts in our objective function. With this, we have:
\begin{eqnarray}
    \sqrt{j^2-(\alpha-j)^2}\sqrt{k'-\alpha}&=&\sqrt{2j\alpha-\alpha^2}\sqrt{k'-\alpha}\\
    &=& \sqrt{\alpha}\sqrt{2j-\alpha}\sqrt{k'-\alpha}
\end{eqnarray}
Solving for the roots again using this simplified expression provides:
\begin{eqnarray}
    \alpha &=&\frac{1}{3}(2j+k'-\sqrt{4j^2-2jk'+(k')^2})+O(\sqrt{j}+\sqrt{k})\\
    &=&\frac{1}{3}(2j+k'-\sqrt{(2j+k')^2-6jk'})\\
    &=&\frac{1}{3}(2j+k'-(2j+k')\sqrt{1-\frac{6jk'}{(2j+k')^2}})\\
    &=& \frac{1}{3}(2j+k')(1-\sqrt{1-\frac{6jk'}{(2j+k')^2}})
\end{eqnarray}

Observe that $\max_{j,k'} \frac{6jk'}{(2j+k')^2}=\frac{3}{4}$ where $k'=2j$, and $\min_{j,k'}\frac{6jk'}{(2j+k')^2}=0$ when one is constant and the other approaches infinity. This means that:
\begin{equation}
    0< \alpha\leq \frac{1}{6}(2j+k')
\end{equation}
Putting this into our expression for the largest eigenvalue, being sure to include the factor of two due to there being a second entry, provides:
\begin{eqnarray}
    \max \Lambda(j,k) &<& 2 \sqrt{\frac{1}{6}(2j+k')}\sqrt{2j}\sqrt{k'}\\
    &=&\frac{2}{\sqrt{3}}\sqrt{(2j+k')jk'}+O(j^{3/4}+(k')^{3/4})
\end{eqnarray}
This expression is mostly relevant when $2j\approx k'$. We now move to the cases of $2j\ll k'$ and $k'\ll 2j$. Returning to our prior expression this is:
\begin{eqnarray}
    & & \sqrt{\frac{1}{3}(2j+k')(1-\sqrt{1-\frac{6jk'}{(2j+k')^2}})}\sqrt{2j-\frac{1}{3}(2j+k')(1-\sqrt{1-\frac{6jk'}{(2j+k')^2}})}\\
    & &\times \sqrt{k'-\frac{1}{3}(2j+k')(1-\sqrt{1-\frac{6jk'}{(2j+k')^2}})}\\
    &=& \sqrt{\frac{2}{27}(8j^3(\sqrt{1-\frac{6jk}{(2j+k)^2}}-1)+6j^2k+k^3(\sqrt{1-\frac{6jk}{(2j+k)^2}}-1)+3jk^2)}
\end{eqnarray}
Taking a series expansion of this and doubling for there being two entries we have:

\begin{eqnarray}
   \max \Lambda(j,k) &\leq & \begin{cases}
        2\sqrt{j^2k'-j^3+\frac{j^4}{4k'}} & 2j\ll k'\\
        2\sqrt{\frac{j(k')^2}{2}-\frac{1}{8}(k')^3+\frac{1}{128}\frac{(k')^4}{j}} & k'\ll 2j
    \end{cases}\\
    &\approx & \begin{cases}
        2[j\sqrt{k'}-\frac{1}{2}\frac{j^2}{\sqrt{k}}+\frac{1}{8}\frac{j^4}{k^{5/2}}+O(j^5/(k')^{7/2})] & 2j\ll k'\\
        2[\frac{1}{\sqrt{2}}k'\sqrt{j}-\frac{1}{8\sqrt{2}}\frac{(k')^2}{\sqrt{j}}+\frac{1}{512}\frac{(k')^4}{j^{5/2}}+O((k')^5/j^{7/2})] & k'\ll 2j
    \end{cases}
\end{eqnarray}

\end{proof}

\begin{proof}[Proof of equation (\ref{degenmoment})]
This relation can be seen as the weighted average of averages, and can thus be derived as follows:
\begin{align*}
    \avg{\Lambda(k)^t} &= \frac{1}{D_k}\sum_{j} d_j \abs{\mathcal{B}_{j,k}} \avg{\Lambda(j,k)^t}\\
    &=\frac{1}{D_k}\sum_{j} d_j \abs{\mathcal{B}_{j,k}} \frac{\Tr(L(j,k)^t}{\abs{\mathcal{B}_{j,k}}}\\
    &=\frac{1}{D_k}\sum_{j} d_j \Tr (L(j,k)^t).
\end{align*}
\end{proof}

\begin{lemma}\label{fraction}
The degeneracies in our system satisfy:
\begin{equation}
    \frac{d_j}{2^N}=O(\frac{1}{N})
\end{equation}
for all allowed values of $j$.
\end{lemma}

\begin{proof}
Since $d_j < d_{j^*}$ for each $j$, we particularize to $j=j^*$. Then, taking only the leading term of $j^* = \sqrt{N}/2$, we have
\begin{equation}
    d_{j^*} = \frac{\sqrt{N}/2 + 1}{N/2 + \sqrt{N}/2 + 1}\binom{N}{N/2 + \sqrt{N}/2 + 1}.
\end{equation}
Focusing on the first factor,
\begin{equation}
    \frac{2}{\frac{N+1}{\sqrt{N} + 1} + 1 } = O(1/\sqrt{N}).
\end{equation}

Since $k = N/2 + \sqrt{N}/2 + 1$, we have that $\abs{N/2 - k} = o(n^{2/3})$, we can utilize the following asymptotic equivalence relation\cite{spencer2014asymptopia}:
\begin{equation}
    \binom{N}{k} \sim \frac{2^N}{\sqrt{N\pi/2}}e^{-(N-2k)^2 / (2N)}.
\end{equation}
Then, using the fact that $N-2k = \sqrt{N} - 2$, we find that 
\begin{align}
    \nonumber e^{-(N-2k)^2 / (2N)} &= e^{-(\sqrt{N}-2)^2/(2N)}\\
    \nonumber &= e^{-1/2 + 2/\sqrt{N} - 2/N}\\
    &= \frac{1}{\sqrt{e}} \big(1 + O(1/\sqrt{N})\big).
\end{align}

Putting together the leading term with the asymptotic equivalence relation, we find
\begin{align}
    \nonumber \binom{N}{N/2 + \sqrt{N}/2 + 1} &\sim \frac{2^N}{\sqrt{N\pi e/2}} \big(1 + O(1/\sqrt{N})\big)\\
    &= O(2^N/\sqrt{N}).
\end{align}
This finally implies
\begin{equation}
    d_{j^*} = O(2^N / N),
\end{equation}
and thus it hold that for all allowed $j$,
\begin{equation}
    d_{j} 2^{-N} =  O(1/N).
\end{equation}
\end{proof}

\begin{proof}[Proof of equation (\ref{degenavg})]
\if{false}To show this result, we first bound the proportion of states in a given angular momentum subspace relative to the whole. That is, for all allowed values of $j$,
\begin{equation}
    \frac{d_j}{2^N} = O\big(\frac{1}{N}\big).
\end{equation}
As proven in above in Lemma \ref{fraction}.\fi

In order to derive our result, we particularize to $k>N$, since this fixes $D_k = 2^N$ and $k > N/2 + j$ is true for each value of $j$. Then,
\begin{equation}
    \operatorname{Var}(\Lambda(k)) = \frac{1}{2^N}\sum_{j} d_j \Tr (L(j,k)^2).  
\end{equation}

Recall the trace of the square of a coupling matrix is given by,
\begin{align}\label{eq:tr_of_sqr}
    \nonumber \Tr(L(j,k)^2) &= \frac{1}{2}\abs{\mathcal{B}_{j,k}}^4 -     \frac{1}{3}\abs{\mathcal{B}_{j,k}}^3(2k'+4j+7)+ \abs{\mathcal{B}_{j,k}}^2(2jk'+2k'+4j+7/2)- \frac{1}{3}\abs{\mathcal{B}_{j,k}}(6jk' + 8j + 4k' + 5)
\end{align}
Using the fact now that $\abs{\mathcal{B}_{j,k}} = 2j+1$ and $k' = k - N/2 + j$, we find that
\begin{align}
    \nonumber \Tr(L(j,k)^2) &= k\big(\frac{8}{3}j^3 + 4j^2 + \frac{4}{3}j\big)- N\big(\frac{4}{3}j^3 + 2j^2 + \frac{2}{3}j\big)+ \big(\frac{4}{3}j^3 + 2j^2 + \frac{2}{3}j\big).
\end{align}
We focus on only the terms of order $k$, thus the dominant part of the expression we wish to analyze is given by
\begin{equation}
\frac{4}{3}k(2j^3 + 3j^2 + j).
\end{equation}
It remains to determine the order $k$ contribution to the entire variance, upon averaging over the degeneracies, 
\begin{equation}
    \operatorname{Var}(\Lambda(k)) = \frac{4}{3} k \sum_{j} \frac{d_j}{2^N} (2j^3 + 3j^2 + j) + \cdots,
\end{equation}
where the terms of order $k^0$ will be dropped moving forward. 
    
In order to make the sum over $j$ tractable, we make use of Lemma \ref{fraction},
\begin{equation}
    \frac{d_j}{2^N} = O\big(\frac{1}{N}\big).
\end{equation}
Given that the strong support of the weighting function is from $0$ to $O(\sqrt{N})$, we have,
\begin{align}
    \nonumber \sum_{j} \frac{4 k d_j}{2^N} \frac{2j^3 + 3j^2 - 2j}{3} &=  \frac{4k}{3} \sum_{j} O\big(\frac{1}{N}\big) (2j^3 + 3j^2 - 2j)\\
    \nonumber &\approx O\big(\frac{k}{N}\big)  \sum_{j=0}^{O(\sqrt{N})}(2j^3 + 3j^2 - 2j)\\
    \nonumber &= O\big(\frac{k}{N}\big) O(N^2)\\
    &= O(Nk).
\end{align}
\end{proof}
\if{false}

\begin{proof}[Proof of equation \eqref{eq:lambda_max_cases}]
Consider the special case of a tridiagonal Toelpitz matrix 
    \begin{equation}
        A_{n}=\left[\begin{array}{cccccc}a & b & 0 & 0 & \cdots & 0 \\ c & a & b & 0 & \cdots & 0 \\ 0 & c & a & b & \cdots & 0 \\ \vdots & \vdots & \vdots & \ddots & \cdots & 0 \\ \vdots & \vdots & \vdots & \vdots & a & b \\ 0 & 0 & 0 & \cdots & c & a\end{array}\right],
    \end{equation}
    with eigenvalues given by \cite{noschese2013tridiagonal}
    \begin{equation}
        \lambda_k = a + 2\sqrt{bc}\cos{\frac{k\pi}{n+1}}.
    \end{equation}
    for $k = 1,2, \dots, n$.
    Now consider the case of a symmetric tridiagonal but not Toelpitz matrix 
    \begin{equation}
        B_{n}=\left[\begin{array}{cccccc}0 & b_1 & 0 & 0 & \cdots & 0 \\ b_1 & 0 & b_2 & 0 & \cdots & 0 \\ 0 & b_2 & 0 & b_3 & \cdots & 0 \\ \vdots & \vdots & \vdots & \ddots & \cdots & 0 \\ \vdots & \vdots & \vdots & \vdots & 0 & b_{n} \\ 0 & 0 & 0 & \cdots & b_{n} & 0\end{array}\right],
    \end{equation}
    with all entries non--negative and each $b_i$ bounded above such that there exists a maximal entry, $b_{max}$, so that $b_i \leq b_{\operatorname{max}}$ for all $i$. Then the eigenvector corresponding to the largest eigenvalue must have entries which are all positive by Perron-Frobenius theorem, and its eigenvalue is an increasing function in $b_i$ by the Gershgorin circle theorem. In the Toelpitz matrix formed with $b_{max}$ on the off diagonals, the largest eigenvalue is $2b_{max}\cos{\frac{1}{n+1}}$.  Thus, for the symmetric tridiagonal matrix the supremum of the eigenvalues is $2b_{\operatorname{max}} \cos{\frac{1}{n+1}}$.
    
    Going back to our coupling matrix, we have that each $l_i(j,k)$ is non--negative and bounded by $j\sqrt{k}$, and $n_{j,k} = \min\{2j, k-k_0(j)\}$. Thus, the supremum of the eigenvalues of the coupling matrix $L(j,k)$ is
    \begin{align}
        \operatorname{sup} \Lambda(j, k) = 
        \begin{cases}
        2j\sqrt{k}\cos{\frac{1}{2(2j+1)}}, & n_{j,k} = 2j+1 \\
        2j\sqrt{k}\cos{\frac{1}{2(k-k_0(j))}}, & n_{j,k} = k-k_0(j).
    \end{cases}
    \end{align}
\end{proof}

\fi
\clearpage

\subsubsection*{N=2 Case of the Tavis--Cummings Model}

Here we solve the Tavis--Cummings model for the case of two spins interacting with a cavity. We show that this problem, beyond the first couple of rungs of the spectrum, can be solved in two different ways: one where we just solve the full matrix for two spin-$1/2$'s and one where we split the matrix into a spin-0 space and spin-1 space.

Keeping with the same convention as before the spectrum can be written graphically as:

\[
\begin{matrix}
&\vdots & & \vdots & & \vdots & & \vdots\\
\vspace{0.3cm}\\
k=3 &\rule{1cm}{0.1cm} & \hspace{1cm} & \rule{1cm}{0.1cm} & \hspace{1cm} & \rule{1cm}{0.1cm} & \hspace{1cm} & \rule{1cm}{0.1cm}\\
\vspace{0.5cm}\\
k=2 &\rule{1cm}{0.1cm} & \hspace{1cm} & \rule{1cm}{0.1cm} & \hspace{1cm} & \rule{1cm}{0.1cm} & \hspace{1cm} & \rule{1cm}{0.1cm}\\
\vspace{0.5cm}\\
k=1 &\rule{1cm}{0.1cm} & \hspace{1cm} & \rule{1cm}{0.1cm} & \hspace{1cm} & \rule{1cm}{0.1cm}\\
\vspace{0.5cm}\\
k=0 &\rule{1cm}{0.1cm} & & \\
\vspace{0.2cm}\\
&|\downarrow\downarrow\rangle & & |\downarrow\uparrow\rangle & & |\uparrow\downarrow\rangle & & |\uparrow\uparrow\rangle
\end{matrix}
\]

Again the first row's state, $|0\rangle|\downarrow\downarrow\rangle$, is unperturbed under the coupling term.

The second row has coupling matrix:
\begin{equation}
\begin{bmatrix}
0 & 1 & 1 \\
1 & 0 & 0 \\
1 & 0 & 0\\
\end{bmatrix}
\end{equation}
We find that the dressed states and perturbed energies are:
\begin{eqnarray}
    \frac{1}{2}[\sqrt{2}|1\rangle|\downarrow\downarrow\rangle+|0\rangle|\downarrow\uparrow\rangle+|0\rangle|\uparrow\downarrow\rangle] &,&\quad E=\omega+g_0\sqrt{2}\\
    \frac{1}{2}[\sqrt{2}|1\rangle|\downarrow\downarrow\rangle-|0\rangle|\downarrow\uparrow\rangle-|0\rangle|\uparrow\downarrow\rangle] &,&\quad E=\omega-g_0\sqrt{2}\\
    \frac{1}{\sqrt{2}}[|0\rangle|\downarrow\uparrow\rangle-|0\rangle|\uparrow\downarrow\rangle] &,&\quad E=\omega+0
\end{eqnarray}
The set of splittings here are $\{0,\pm g_0\sqrt{2}\}$.

For the further levels with $k$ excitations, $k\geq 2$, there are always exactly four states that we must diagonalize, providing the coupling matrix:
\begin{equation}
\begin{bmatrix}
0 & \sqrt{k} & \sqrt{k} & 0 \\
\sqrt{k} & 0 & 0 & \sqrt{k-1} \\
\sqrt{k} & 0 & 0 & \sqrt{k-1} \\
0 & \sqrt{k-1} & \sqrt{k-1} & 0\\
\end{bmatrix}
\end{equation}
Solving this directly the dressed states and energies are:

\begin{eqnarray}
    \frac{1}{\sqrt{2k-1}}[-\sqrt{k-1}|k\rangle|\downarrow\downarrow\rangle+\sqrt{k}|k-2\rangle|\uparrow\uparrow\rangle] &,&\\
    E=k\omega_0+0 & &\\
    \frac{1}{\sqrt{2}}[|k-1\rangle|\downarrow\uparrow\rangle-|k-1\rangle|\uparrow\downarrow\rangle] &,&\\
    E=k\omega_0+0 & &\\
    \frac{1}{2\sqrt{2k-1}}[\sqrt{2}\sqrt{k}|k\rangle|\downarrow\downarrow\rangle+\sqrt{2k-1}|k-1\rangle|\downarrow\uparrow\rangle+\sqrt{2k-1}|k-1\rangle|\uparrow\downarrow\rangle+\sqrt{2}\sqrt{k-1}|k-2\rangle|\uparrow\uparrow\rangle]&,&\\
    E=k\omega_0+g_0\sqrt{2}\sqrt{2k-1}& &\\
    \frac{1}{2\sqrt{2k-1}}[\sqrt{2}\sqrt{k}|k\rangle|\downarrow\downarrow\rangle-\sqrt{2k-1}|k-1\rangle|\downarrow\uparrow\rangle-\sqrt{2k-1}|k-1\rangle|\uparrow\downarrow\rangle+\sqrt{2}\sqrt{k-1}|k-2\rangle|\uparrow\uparrow\rangle]&,&\\
    E=k\omega_0-g_0\sqrt{2}\sqrt{2k-1}& &
\end{eqnarray}

This provides a complete description of the states. We carried out this computation with the Zeeman basis states. If we instead take as our basis states the collective spin states with constant angular momentum (spin-1 and spin-0), we slightly reduce the complexity of the problem. In this case the only difference is replacing $|k\rangle|\uparrow\downarrow\rangle$ and $|k\rangle|\downarrow\uparrow\rangle$ with $|k\rangle (|\uparrow\downarrow\rangle+|\downarrow\uparrow\rangle)$ and $|k\rangle (|\uparrow\downarrow\rangle-|\downarrow\uparrow\rangle)$--the first breaks off a spin-1 space while the second breaks off a spin-0 space.

In this frame, the singlet state $|\uparrow\downarrow\rangle-|\downarrow\uparrow\rangle$ is always annihilated by $J_{\pm}$, thus forming its own space. Generally this particular decomposition doesn't work, but the space is still orthogonal to the other states. In this case since this singlet state removes one vector, the coupling matrices for the above becomes smaller and thus easier to diagonalize. The resulting states and energies will of course be left unchanged.

\subsubsection*{N=3 Case of the Tavis--Cummings Model with $k\leq 2$}

For the sake of completeness we show the dressed states for $k\leq 2$ for the $N=3$ case of the Tavis--Cummings model, which were neglected in Section \ref{N3}. As always the $k=0$ state is unperturbed, giving $|0\rangle|\downarrow\downarrow\downarrow\rangle$ with energy $E=0$.

At $k=1$ there are four bases, which in terms of the standard bases gives a coupling matrix of:
\begin{equation}
    g\begin{bmatrix}
    0 & 1 & 1 & 1\\
    1 & 0 & 0 & 0\\
    1 & 0 & 0 & 0\\
    1 & 0 & 0 & 0
    \end{bmatrix}
\end{equation}
Solving this matrix, we find that the dressed states and their associated energies are given by:
\begin{eqnarray}
    \frac{1}{\sqrt{6}}[\sqrt{3}|1\rangle|\downarrow\downarrow\downarrow\rangle+|0\rangle(|\downarrow\downarrow\uparrow\rangle+|\downarrow\uparrow\downarrow\rangle+|\uparrow\downarrow\downarrow\rangle)]&,&\quad E=\omega_0+g_0\sqrt{3}\\
    \frac{1}{\sqrt{6}}[-\sqrt{3}|1\rangle|\downarrow\downarrow\downarrow\rangle+|0\rangle(|\downarrow\downarrow\uparrow\rangle+|\downarrow\uparrow\downarrow\rangle+|\uparrow\downarrow\downarrow\rangle)]&,&\quad E=\omega_0-g_0\sqrt{3}\\
    \frac{1}{\sqrt{2}}|0\rangle(|\uparrow\downarrow\downarrow\rangle-|\downarrow\downarrow\uparrow\rangle)&,&\quad E=\omega_0\\
    \frac{1}{\sqrt{2}}|0\rangle(|\downarrow\uparrow\downarrow\rangle-|\downarrow\downarrow\uparrow\rangle)&,&\quad E=\omega_0
\end{eqnarray}
At $k=2$ there are now seven bases, which in terms of the standard Zeeman bases gives a coupling matrix of:
\begin{equation}
\begin{bmatrix}
0 & \sqrt{2} & \sqrt{2} & 0 & \sqrt{2} & 0 & 0 \\
\sqrt{2} & 0 & 0 & 1 & 0 & 1 & 0 \\
\sqrt{2} & 0 & 0 & 1 & 0 & 0 & 1 \\
0 & 1 & 1 & 0 & 0 & 0 & 0 \\
\sqrt{2} & 0 & 0 & 0 & 0 & 1 & 1 \\
0 & 1 & 0 & 0 & 1 & 0 & 0 \\
0 & 0 & 1 & 0 & 1 & 0 & 0 \\
\end{bmatrix}
\end{equation}
where the basis states are $\{|k\rangle |\downarrow\downarrow\downarrow\rangle,|k-1\rangle |\downarrow\downarrow\uparrow\rangle,\ldots |k-3\rangle |\uparrow\uparrow\uparrow\rangle\}$. This can be decomposed as:
\begin{equation}
    \begin{bmatrix}
    0 & \sqrt{6} & 0 \\
    \sqrt{6} & 0 & 2 \\
    0 & 2 & 0 
    \end{bmatrix}\oplus\begin{bmatrix}
    0 & 1\\
    1 & 0
    \end{bmatrix}\oplus\begin{bmatrix}
    0 & 1\\
    1 & 0
    \end{bmatrix}
\end{equation}
The first matrix in the direct sum has dressed states:
\begin{eqnarray}
    \frac{1}{2\sqrt{3}}[\sqrt{\frac{3}{2}}|2\rangle |\downarrow\downarrow\downarrow\rangle -\sqrt{\frac{5}{2}}|1\rangle (|\downarrow\downarrow\uparrow\rangle+|\downarrow\uparrow\downarrow\rangle+|\uparrow\downarrow\downarrow\rangle)+|0\rangle (|\downarrow\uparrow\uparrow\rangle+|\uparrow\downarrow\uparrow\rangle+|\uparrow\uparrow\downarrow\rangle)]&,&\ E=2\omega_0-g_0\sqrt{10}\\
    \frac{1}{2\sqrt{3}}[\sqrt{\frac{3}{2}}|2\rangle |\downarrow\downarrow\downarrow\rangle +\sqrt{\frac{5}{2}}|1\rangle (|\downarrow\downarrow\uparrow\rangle+|\downarrow\uparrow\downarrow\rangle+|\uparrow\downarrow\downarrow\rangle)+|0\rangle (|\downarrow\uparrow\uparrow\rangle+|\uparrow\downarrow\uparrow\rangle+|\uparrow\uparrow\downarrow\rangle)]&,&\ E=2\omega_0+g_0\sqrt{10}\\
    \sqrt{\frac{3}{11}}[-\sqrt{\frac{2}{3}}|2\rangle |\downarrow\downarrow\downarrow\rangle +|0\rangle (|\downarrow\uparrow\uparrow\rangle+|\uparrow\downarrow\uparrow\rangle+|\uparrow\uparrow\downarrow\rangle)]&,&\ E=2\omega_0
\end{eqnarray}
The second and third matrices can be diagonalized in the following bases:
\begin{eqnarray}
\frac{1}{2}[|1\rangle[|\downarrow\uparrow\downarrow\rangle-|\uparrow\downarrow\downarrow\rangle]\pm|0\rangle [|\downarrow\uparrow\uparrow\rangle-|\uparrow\downarrow\uparrow\rangle]]&,&\quad E=2\omega_0 \pm g_0\\
\frac{1}{2\sqrt{3}}[|1\rangle[2|\downarrow\downarrow\uparrow\rangle-|\uparrow\downarrow\downarrow\rangle-|\downarrow\uparrow\downarrow\rangle]\pm |0\rangle [|\uparrow\downarrow\uparrow\rangle+|\downarrow\uparrow\uparrow\rangle-2|\uparrow\uparrow\downarrow\rangle]]&,&\quad E=2\omega_0 \pm g_0
\end{eqnarray}
Combining these results with those from Section \ref{N3} provides the full set of dressed states and re-diagonalized energies for the Tavis--Cummings model at $N=3$.

\end{widetext}
\end{document}